\documentclass[journal]{IEEEtran}

\usepackage{graphicx}
\usepackage{amsmath, amsfonts}
\usepackage{amsthm}
\usepackage{mathtools}
\usepackage{amssymb}
\usepackage{lscape}
\usepackage{verbatim}
\usepackage[caption=false,font=normalsize,labelfont=sf,textfont=sf]{subfig}
\usepackage{multirow}
\usepackage{algpseudocode}
\usepackage[linesnumbered,ruled,vlined]{algorithm2e}

\usepackage{array}
\usepackage{textcomp}
\usepackage{stfloats}
\usepackage{url}
\usepackage{balance}
\usepackage{cite}
\usepackage{xcolor}

\usepackage{enumerate}

\newtheorem{myDef}{Definition}
\newtheorem{myTheor}{Theorem}

\newtheorem{myLemma}{Lemma}
\newtheorem{myExp}{Example}
\newtheorem{myProposition}{Proposition}
\newtheorem{myCorollary}{Corollary}

\newcommand{\erfc}{\mathrm{erfc}}

\begin{document}

\title{Finite Dimensional Lattice Codes with \\ Self Error-Detection and Retry Decoding}

\author{Jiajie Xue,~\IEEEmembership{Graduate student member,~IEEE,} and Brian M. Kurkoski,~\IEEEmembership{Member,~IEEE}
\thanks{The material in this paper was presented in part at the 2022 IEEE International Symposium on Information Theory \cite{xue2022lower}.}
\thanks{The author are with the Graduate School of Advanced Science and Technology, Japan Advanced Institute of Science and Technology, Nomi 923-1292, Japan (e-mail: xue.jiajie, kurkoski@jaist.ac.jp).}
\thanks{This work was supported by JSPS Kakenhi Grant Number JP 21H04873.}}


\maketitle

\begin{abstract}
Lattice codes with optimal decoding coefficient are capacity-achieving when dimension $N \rightarrow \infty$. In communications systems, finite dimensional lattice codes are considered, where the optimal decoding coefficients may still fail decoding even when $R< C$. This paper presents a new retry decoding scheme for finite dimensional lattice-based transmissions. When decoding errors are detected, the receiver is allowed to adjust the value of decoding coefficients and retry decoding, instead of requesting a re-transmission immediately which causes high latency. This scheme is considered for both point-to-point single user transmission and compute-forward (CF) relaying with power unconstrained relays, by which a lower word error rate (WER) is achieved than conventional one-shot decoding with optimal coefficients. A lattice/lattice code construction, called CRC-embedded lattice/lattice code, is presented to provide physical layer error detection to enable retry decoding. For CF relaying, a shaping lattice design is given so that the decoder is able to detect errors from CF linear combinations without requiring individual users' messages. The numerical results show gains of up to 1.31 dB and 1.08 dB at error probability $10^{-5}$ for a 2-user CF relay using 128- and 256-dimensional lattice codes with optimized CRC length and 2 decoding trials in total.

\end{abstract}


\section{Introduction}  \label{sec_introduction}
\IEEEPARstart{A}{} lattice $\Lambda$ is a discrete additive subgroup of the real number space $\mathbb{R}^N$. Since $\Lambda$ is an infinite constellation and is power unconstrained, a lattice code $\mathcal{C}= \Lambda_c/ \Lambda_s$ may be constructed using a coding lattice $\Lambda_c$ and a shaping lattice $\Lambda_s$ to satisfy power constraints for wireless communications. Lattice codes can be seen as a coded modulation scheme for power-constrained communication systems, which has shaping gain to reduce transmission power. Theoretical results on asymptotic dimensional lattice codes show that lattice codes are capacity-achieving when using maximum likelihood (ML) decoding by choosing shaping region as a $N$-dimensional thin shell \cite{de1989some} \cite{linder1993corrected} and later extended to the whole $N$-dimensional sphere \cite{urbanke1998lattice}. More significantly, it is shown that the capacity can also be achieved using low complexity lattice decoding if the received message is scaled by an MMSE factor, denoted as $\alpha_{MMSE}$, before decoding \cite{erez2004achieving}. Besides these results for the asymptotic case, properties of practical and finite dimensional lattices/lattice codes are also widely studied. 
Conway and Sloane's book \cite{conway1993sphere} describes a series of well-known low dimensional lattices. For higher dimensional lattices, recent research interests are on designing structured lattice codes for wireless communications, such as low density lattice codes \cite{sommer2008low}, Construction D/D' lattice codes based on BCH codes \cite{matsumine2018construction}, polar codes \cite{liu2018construction} \cite{ludwiniananda2021design} and LDPC codes \cite{zhou2022construction}, which give the excellent error performance and low decoding complexity. 

Because of their linearity, lattice codes are suitable for physical layer network coding (PLNC). Compute-forward (CF) relaying \cite{nazer2011compute} is a multiple access relaying technique which utilizes the linearity of lattice codes for PLNC and can significantly improve the network throughput. Instead of multi-user detection, CF relay applies a single-user decoder to estimate a linear combination of users' messages. Previous research studied the performance on various of network topologies in \cite{zhu2016gaussian} \cite{hasan2017practical} from theoretical perspective and code designs using binary LDPC codes \cite{sula2018compute} and Construction A lattices \cite{ordentlich2011practical}. 

\IEEEpubidadjcol

\subsection{Problem statement and motivation}
This paper considers finite-dimensional lattice-based communications for point-to-point single user transmission and CF relaying. In theoretical studies, lattice decoding with scaling factor $\alpha_{MMSE}$ is capacity achieving when the dimension is asymptotically large. For finite dimensional lattice codes with $R< C$, the probability of decoding error is non-zero and $\alpha_{MMSE}$ may still fail decoding. Conventionally, the receiver requests a re-transmission for a failed decoding, which causes high latency. 
A similar situation is faced by CF relaying. A decoding coefficient set $\{\mathbf{a}, \alpha\}$ is selected at CF relay, which consists of an integer vector $\mathbf{a}$ as coefficients of the linear combination and a scaling factor $\alpha$. A decoding error happens when the relay cannot decode a correct linear combination with given coefficient set $\{\mathbf{a}, \alpha\}$. Since a linear combination includes multiple users' messages, a stand-alone relay may not be able to perform error detection from one-shot decoding and may forward error-containing messages into network causing decoding failure at the destination.

The motivation of this paper is to investigate that, when optimal coefficient(s) for single user case or CF relaying fail decoding, if the receiver can improve the error performance using retry decoding after adjusting the value of decoding coefficients without requesting re-transmission. In order to implement retry decoding and prevent forwarding erroneously decoded messages, a lattice construction is also studied which provides physical layer error detection ability and is applicable to both single user case and CF relaying.
For CF relaying, the lattice construction should provide functional error detection for linear combinations at a stand-alone relay even without knowledge of individual users' messages.

\subsection{Contributions}
The contributions of this paper are summarized into two parts. First, in order to improve error performance, we give a retry decoding scheme which adjusts the values of decoding coefficients after the current coefficients have failed, for single user (SU) transmission and CF relaying. We show that, even though $\alpha_{MMSE}$ \cite{erez2004achieving} and CF coefficient set $\{\mathbf{a}, \alpha\}$ derived in \cite{nazer2011compute} give optimal error performance for one-shot decoding, space for improvement still exists by using retry decoding, especially for low, e.g. dimension $N= 8$, and medium dimensional, e.g. $N= 128$, lattice codes. The coefficient candidates are listed based on probability of correct decoding given all previous candidates failed for the SU case; or computation rate for CF relaying. This ensures the coefficient candidates are tested in the order of reliability to reduce the number of retries. For the SU case, an offline algorithm is given to find a finite-length candidate list by using a genie-aided exhaustive search decoder. Since the candidate search algorithm is performed offline, the complexity of exhaustive search does not affect the implementation. A lower bound on error probability is derived for this decoder by extending the finite-length list to the set of all real numbers. 

Second, we propose a new lattice construction technique which adds physical layer error detection to any existing lattices. Error detection is implemented by restricting the least significant bits (LSB) of lattice uncoded messages $\mathbf{b}_{LSB}$ using a binary linear block code $\mathcal{C}_b$. 
For practical design of $\mathcal{C}_b$, cyclic redundancy check (CRC) codes, which are widely used for error detection in systems, are mainly considered in this paper, named CRC-embedded lattice/lattice code.
The construction of CRC-integrated lattice codes for error detection, to the best of our knowledge, has not been studied. The CRC-embedded lattice code is valid for both SU transmission and CF relaying. For CF relaying, a condition on shaping lattice design needs to be satisfied in order to detect errors from linear combinations without knowledge of individual users' messages. The error detection capability of the embedded CRC code is evaluated by the probability of undetected error with respect to CRC length $l$. As the number of CRC parity bits increases, the CRC-embedded lattice code has better error detection capability while a larger SNR penalty is suffered. To balance this trade-off, CRC length optimization is given to maximize the SNR gain for a target error rate, which is semi-analytical and does not require a search over the CRC length. 
An implementation of CRC-embedded lattice codes with retry decoding is given for the SU case, using $E_8$ and $BW_{16}$ lattice codes, and CF relaying, using construction D polar code lattice with dimension $N= 128, 256$. The benefit of retry decoding is illustrated along with the optimized CRC length. A more significant gain is observed for CF relaying than the SU case. For a 2-user CF relay, 1.31 dB and 1.08 dB gain are achieved for equation error rate (EER) of CF linear combination being $10^{-5}$ by only adding one more decoding attempt when $N= 128, 256$, respectively. 

The organization of this paper is as follows. Section~\ref{sec_preliminary} reviews definitions, theoretical results and system models of lattice/lattice code transmission and CF relaying. Section~\ref{sec_su} and Section~\ref{sec_mac_CFrelay} describe the retry decoding scheme for single user transmission and CF relaying, respectively. In Section~\ref{sec_su} and~\ref{sec_mac_CFrelay}, we assume genie-aided error detection, for which the true message is known at decoder but only used for error detection. Section~\ref{sec_code_construct} gives the construction of the CRC-embedded lattice/lattice codes to provide error detection in practical decoding. {{} Section~\ref{sec_CRC_length_opt} gives the optimization of the CRC length. Section~\ref{sec_implement} gives numerical results on implementation of the CRC-embedded lattice codes with retry decoding and optimized CRC length.} Finally, Section~\ref{sec_conclusions} gives the conclusions of this paper with discussions of extension of this work.

\subsection{Notations}
Notations used in this paper are described. Scalar variables are denoted using italic font, e.g. code rate $R$ and lattice dimension $N$; vectors are denoted using lower-case bold, e.g. message vector $\mathbf{x}, \mathbf{y}$; matrices are denoted using upper-case bold, e.g. generator matrix $\mathbf{G}$. Vectors are column vectors, unless stated otherwise. The set of integers and real numbers are denoted using $\mathbb{Z}$ and $\mathbb{R}$, respectively. And an $N$-by-$N$ identity matrix is denoted as $\mathbf{I}_N$. 


\section{Preliminaries} \label{sec_preliminary}
This section gives the definitions of lattices and nested lattice codes which are used for channel coding scheme, and an overview of compute-forward relaying.

\subsection{Lattices}
\begin{myDef} \label{def_lattice}
    \rm (Lattice) An $N$-dimensional lattice $\Lambda$ is a discrete additive subgroup of the real number space $\mathbb{R}^N$.
Let $\mathbf{g}_1, \mathbf{g}_2, ... \mathbf{g}_N \in \mathbb{R}^N$ be $N$ linearly independent column vectors. The lattice $\Lambda$ is formed using generator matrix $\mathbf{G}= [\mathbf{g}_1, \mathbf{g}_2, ... \mathbf{g}_N] \in \mathbb{R}^{N \times N}$ by:
\begin{align}
    \Lambda= \{\mathbf{Gb} | \mathbf{b} \in \mathbb{Z}^N\}.
\end{align}
\end{myDef}
The lattice quantizer $Q_{\Lambda}(\cdot)$ finds the closest lattice point for given $\mathbf{y} \in \mathbb{R}^N$ as:
\begin{align} \label{equ_lattice_quantizer}
    \hat{\mathbf{x}} = Q_{\Lambda}(\mathbf{y}) & = \mathop{\arg\min}_{\mathbf{x} \in \Lambda} \|\mathbf{y}- \mathbf{x}\|^ 2.
\end{align}
In communications, $Q_{\Lambda}(\cdot)$ behaves as the lattice decoder, so we also denote it as
\begin{align}
    \hat{\mathbf{x}}= DEC_{\Lambda}(\mathbf{y}).
\end{align}
Lattice modulo $\bmod\ \Lambda$ is given as:
\begin{align} \label{equ_lattice_mod}
        \mathbf{y} \bmod \Lambda & = \mathbf{y}- Q_{\Lambda}(\mathbf{y}). 
\end{align}

\begin{myDef}
    \rm (Fundamental region) Given a lattice $\Lambda$, a region $\mathcal{F}$ is a fundamental region of $\Lambda$ if: for any distinct lattice points $\mathbf{x}_i \neq \mathbf{x}_j$, $(\mathcal{F}+ \mathbf{x}_i) \cap (\mathcal{F}+ \mathbf{x}_j)= \emptyset$ and the real number space can be covered as $\mathbb{R}^N= \bigcup_{\mathbf{x} \in \Lambda} \mathcal{F}+ \mathbf{x}$. 
\end{myDef}
The fundamental region of a lattice is not unique. The Voronoi region $\mathcal{V}$ and hypercube region $\mathcal{H}$ are considered in this paper and defined respectively as
\begin{align}
    \mathcal{V}(\mathbf{x})= \{\mathbf{y} \in \mathbb{R}^N\ |\ Q_{\Lambda}(\mathbf{y})= \mathbf{x}\},
\end{align}
and, for $\mathbf{G}$ having triangular form with diagonal elements $g_{i, i}$ and index $i=1, 2, \cdots, N$,
\begin{align}
    \mathcal{H}(\mathbf{x})= \{\mathbf{y} \in \mathbb{R}^N | -\frac{g_{i,i}}{2}+ x_{i} \leq y_i < \frac{g_{i,i}}{2}+ x_{i}\}.
\end{align}
The hypercube quantizer $Q_{\mathcal{H}}(\cdot)$ is defined to find the lattice point with respect to $\mathcal{H}$ that $\mathbf{y} \in \mathbb{R}^N$ belongs to. The volume of fundamental region $\mathcal{F}$ is given as:
\begin{align} \label{equ_V_fundamental_region}
    V(\Lambda)= |\det(\mathbf{G})|,
\end{align}
which is independent of the shape of $\mathcal{F}$.

The covering sphere and effective sphere of lattices are defined with respect to the Voronoi region $\mathcal{V}$. The covering sphere $\mathcal{S}_c$ with radius $r_c$ is the sphere of minimal radius that can cover the whole Voronoi region $\mathcal{V}$, that is $\mathcal{V} \subseteq \mathcal{S}_c$. Note that $\mathcal{V} \subset \mathcal{S}_c$ is satisfied for finite dimensional lattices.
The effective sphere $\mathcal{S}_e$ with radius $r_e$ is the sphere having volume $V(\mathcal{S}_e)= V(\Lambda)$, where the volume of the $N$-sphere $\mathcal{S}_e$ is
\begin{align}
    V(\mathcal{S}_e)= \frac{\pi^{N/ 2} r_e^N}{\Gamma(\frac{N}{2}+ 1)},
\end{align}
and $\Gamma(\cdot)$ is the gamma function. The relationship among $\mathcal{V}$, $\mathcal{S}_c$ and $\mathcal{S}_e$ for $N=2$ is shown in Fig.~\ref{fig_eff_cov_sphere}.

\begin{figure}[t]
    \centering
    \includegraphics[scale=0.5]{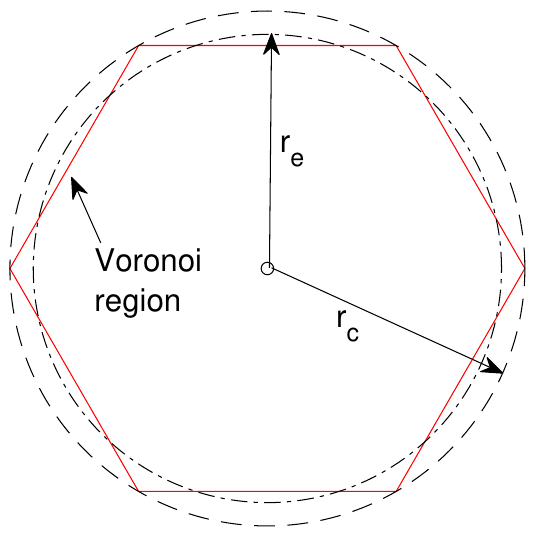}
    \caption{Relationship among lattice Voronoi region, covering sphere (dashed line) and effective sphere (dotted-dashed line).}
    \label{fig_eff_cov_sphere}
\end{figure}

\subsection{Nested lattice codes} \label{sec_nest_lattice}
A lattice $\Lambda$ is an infinite set and does not satisfy the power constraint needed for practical wireless communications. Next we give a review of nested lattice codes (or for short lattice codes), which have a power constraint. 
\begin{myDef}
    \rm (Nested lattice codes) Let two lattices $\Lambda_c$ and $\Lambda_s$ satisfy $\Lambda_s \subseteq \Lambda_c$ and form a quotient group $\Lambda_c / \Lambda_s$. Using a fundamental region $\mathcal{F}_s$ of $\Lambda_s$, a nested lattice code $\mathcal{C}= \Lambda_c/ \Lambda_s$ is:
    \begin{align} \label{equ_def_nestlattice}
        \mathcal{C}= \Lambda_c \cap \mathcal{F}_s.
    \end{align}
\end{myDef}
The fine lattice $\Lambda_c$ is called the coding lattice and the coarse lattice $\Lambda_s$ is called the shaping lattice. Let $\mathbf{G}_c$ and $\mathbf{G}_s$ be generator matrices of $\Lambda_c$ and $\Lambda_s$, respectively. Encoding $\mathcal{C}$ is to map an uncoded message $\mathbf{b} \in \mathbb{Z}^N$ to the codebook $\mathcal{C}$ as:
\begin{align} \label{equ_latticeC_enc}
    \mathbf{x}= ENC(\mathbf{b})= \mathbf{G}_c \mathbf{b} \bmod \Lambda_s= \mathbf{G}_c \mathbf{b}- Q_{\Lambda_s}(\mathbf{G}_c \mathbf{b}).
\end{align}
The $Q_{\Lambda_s}(\cdot)$ in \eqref{equ_latticeC_enc} is a lattice quantizer of $\Lambda_s$ for which there exists an $\mathbf{s} \in \mathbb{Z}^N$ such that $Q_{\Lambda_s}(\mathbf{G}_c \mathbf{b})= \mathbf{G}_s \mathbf{s}$. 
For a nested lattice code $\mathcal{C}= \Lambda_c/ \Lambda_s$, the generator $\mathbf{G}_c$ and $\mathbf{G}_s$ must satisfy the following lemma.
\begin{myLemma} \label{lemma_exist_M}
    \rm \cite[Chapter 8]{zamir2014lattice} A nested lattice code $\mathcal{C}= \Lambda_c/ \Lambda_s$ can be formed if and only if $\mathbf{M}= \mathbf{G}_c^{-1} \mathbf{G}_s$ is a matrix of integers.
\end{myLemma}

A lattice code is constructed by selecting a coding lattice $\Lambda_c$ and a shaping lattice $\Lambda_s$. 
A technique named rectangular encoding \cite{kurkoski2018encoding} gives a method for lattice code design which allows selecting $\Lambda_c$ and $\Lambda_s$ separately for good error performance and low encoding/decoding complexity. By rectangular encoding, the each coordinate of uncoded messages can be defined independently over $b_i \in \{0, 1, 2, \cdots, M_i-1\}$ with positive integers $M_i$ for $i= 1, 2, \cdots, N$, rather than defining a same domain as in the self-similar lattice codes $\Lambda_c/ M\Lambda_c$. To have equal power allocation on each dimensions, we apply hypercube shaping for numerical simulations in this paper, which has equal edge length of shaping region at all dimensions and is applicable for the rectangular encoding. 
For a lattice code, the code rate and average per-dimensional power are defined as follows.
\begin{myDef}
    \rm (Code rate) The code rate of the nested lattice code $\mathcal{C}= \Lambda_c/ \Lambda_s$ with generator matrix $\mathbf{G}_c$ and $\mathbf{G}_s$, and $b_i \in \{0, 1, 2, \cdots, M_i-1\}$ for $i= 1, 2, \cdots, N$ is
    \begin{align} \label{equ_def_latticeR}
        R= \frac{1}{N} \log_2 \frac{|\det(\mathbf{G}_s)|}{|\det(\mathbf{G}_c)|}= \frac{1}{N} \log_2 \prod_{i= 1}^N M_i   (\rm{bits/dimension}).
    \end{align}
\end{myDef} 
\begin{myDef}
    \rm (Average power) The average per-dimensional power of a lattice code $\mathcal{C}$ is 
    \begin{align} \label{equ_average_power}
        P= \frac{1}{N \cdot 2^{NR}} \sum_{\mathbf{x} \in \mathcal{C}} \|\mathbf{x}\|^2.
    \end{align}
\end{myDef}

\subsection{System model for single user transmission} \label{sec_system_model_SU}
For single user transmission, the additive white Gaussian noise (AWGN) channel is considered in this paper. Let codeword $\mathbf{x} \in \mathcal{C}$ have average per-dimensional power $P$. The received message is 
\begin{align}
    \mathbf{y}= \mathbf{x}+ \mathbf{z},
\end{align}
where the Gaussian noise $\mathbf{z} \sim \mathcal{N}(\mathbf{0}, \sigma^2 \mathbf{I}_N)$. The signal-to-noise ratio (SNR) is defined as:
\begin{align} \label{equ_def_SNR}
    SNR= P/ \sigma^2.
\end{align}
The estimate of $\mathbf{x}$ is obtained by a lattice decoder of $\Lambda_c$ with a scaling factor $\alpha \in \mathbb{R}$ as:
\begin{align} \label{equ_lattice_dec}
    \hat{\mathbf{x}}= DEC_{\Lambda_c}(\alpha \mathbf{y}).
\end{align}
The word error rate (WER), that is the ratio of $\hat{\mathbf{x}}= \mathbf{x}$, is evaluated in numerical simulations.
In one-shot decoding, choose $\alpha$ to be the MMSE optimal coefficient \cite{erez2004achieving}:
\begin{align} \label{equ_alpha_MMSE}
    \alpha_{MMSE}= \frac{P}{P+ \sigma^2}. 
\end{align} 
The uncoded message is recovered by taking the inverse of \eqref{equ_latticeC_enc}, called the indexing function:
\begin{align} \label{equ_latticecode_index}
    \hat{\mathbf{b}}= \mathrm{index} (\hat{\mathbf{x}}).
\end{align}
Algorithms for encoding and indexing can be found in \cite[Section IV, V]{kurkoski2018encoding}. 

\subsection{Overview of Compute-Forward} \label{sec_overview_CF}
\begin{figure}
    \centering
    \includegraphics[width=0.9\linewidth]{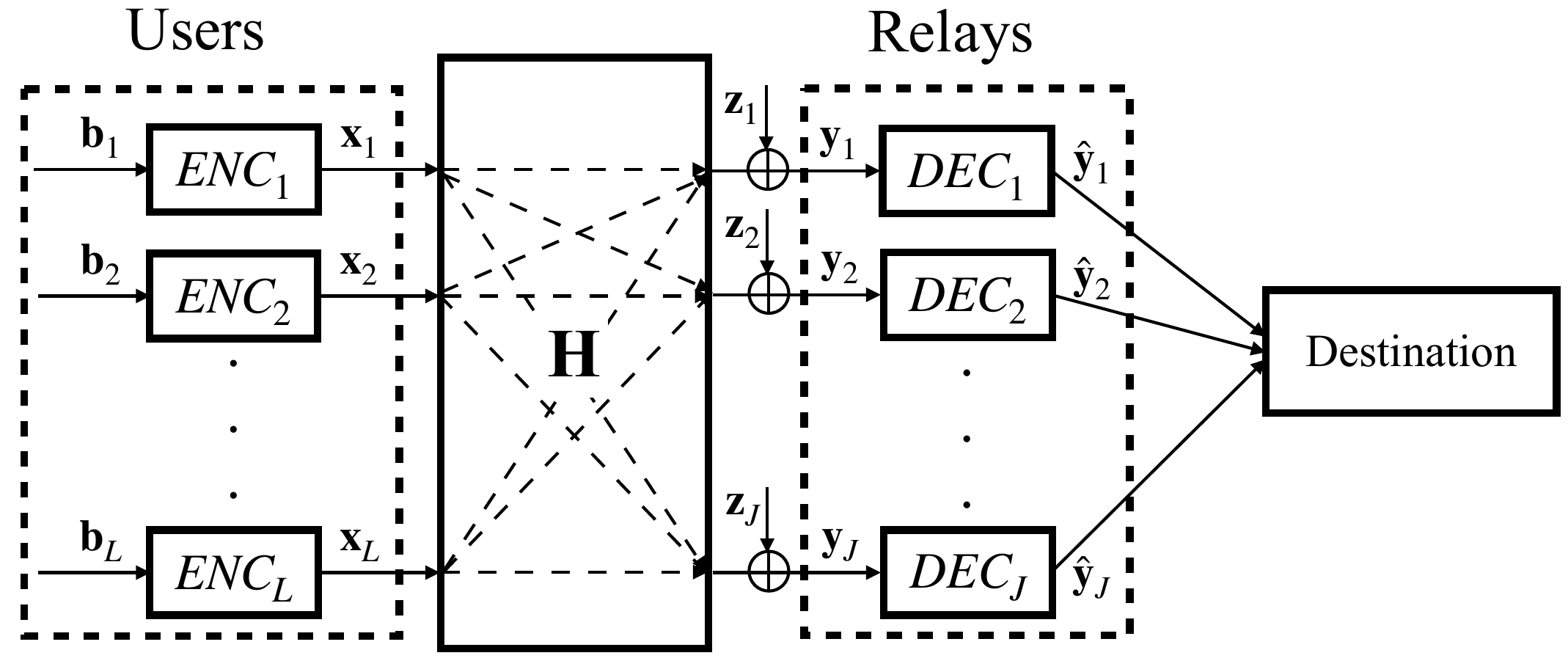}
    \caption{System model of multiple access network compute-forward with $L$ users and $J$ relays.}
    \label{fig_multiple_access_channel}
\end{figure}

Compute-forward is a multiple access relaying scheme \cite{nazer2011compute} where one or more relays aim to decode one or more linear combinations (or linear equations) of users' messages instead of decoding them individually. In \cite{nazer2011compute}, the uncoded message $\mathbf{b} \in \mathbb{F}_q^N$, where $q$ is a prime number. For a real-valued $L$-user $J$-relay system with $J \geq L$ as shown in Fig.~\ref{fig_multiple_access_channel}, suppose all users sharing the same codebook $\mathcal{C}= \Lambda_c/ \Lambda_s$, the users send lattice codewords $\mathbf{x}_i= ENC_i(\mathbf{b}_i)$, for $i= 1, 2, \cdots, L$, through a multiple access channel. The received message at the $j$-th relay is
\begin{align} \label{equ_CF_j_relay_rece}
    \mathbf{y}_j= \sum_{i= 1}^L h_{i, j} \mathbf{x}_i+ \mathbf{z}_j,
\end{align}
where $h_{i, j} \in \mathbb{R}$ is the channel coefficient between the $i$-th user and the $j$-th relay and $\mathbf{z}_j \sim \mathcal{N}(0, \sigma^2 \mathbf{I}_N)$. 
A single-user decoder estimates a desired linear combination $\sum_{i= 1}^L a_{j, i} \mathbf{x}_i \bmod \Lambda_s$ by 
\begin{align} \label{equ_def_latticeequation}
    \hat{\mathbf{y}}_j= DEC_{\Lambda_c}(\alpha_j \mathbf{y}_j) \bmod \Lambda_s,
\end{align}
with $a_{j, i} \in \mathbb{Z}$ and $\alpha_j \in \mathbb{R}$. Then relays forward $\hat{\mathbf{y}}_j$ and $\mathbf{a}_j= [a_{j, 1}, a_{j, 2}, \cdots, a_{j, L}]^T$, for $j= 1, 2, \cdots, J$, to the destination. The users' messages $\hat{\mathbf{b}}_1, \hat{\mathbf{b}}_2, \cdots, \hat{\mathbf{b}}_L$ can be recovered by solving linear equations, if and only if the integer coefficient matrix $([\mathbf{a}_1, \mathbf{a}_2 \cdots, \mathbf{a}_J] \bmod q)$ has rank $L$ over $\mathbb{F}_q$. 

The coefficient set $\{\mathbf{a}, \alpha\}$ is selected as follows. Assume all users share the same codebook $\mathcal{C}$ with power $P$.
Given channel vector $\mathbf{h} \in \mathbb{R}^L$ and $\mathbf{a} \in \mathbb{Z}^L$, the computation rate is defined as
\begin{align} \label{equ_def_computationR}
    R_c(\mathbf{h}, \mathbf{a})= \mathop{\max}_{\alpha \in \mathbb{R}} \frac{1}{2} \log^+ \left( \frac{P}{\alpha^2 \sigma^2+ P \|\alpha \mathbf{h} - \mathbf{a}\|^2} \right).
\end{align}
The values of $\mathbf{a}$ and $\alpha$ are selected to maximize the computation rate $R_c$ as
\begin{align} 
    \mathbf{a} & = {\underset{\mathbf{a}}{\arg\max}} \frac{1}{2} \log^+ \left( \left( \|\mathbf{a}\|^2 - \frac{P(\mathbf{h}^T \mathbf{a})^2}{\sigma^2+ P \|\mathbf{h}\|^2} \right)^{-1} \right), \label{equ_def_optA} \\
    \alpha & = \frac{P \mathbf{h}^T \mathbf{a}}{\sigma^2+ P \|\mathbf{h}\|^2}, \label{equ_def_optalpha}
\end{align}
where $\log^+ (x) \dot= \max(\log(x), 0)$. The integer coefficients $\mathbf{a}$ are restricted by 
\begin{align} \label{equ_CF_a_condi}
    0< \|\mathbf{a}\|^2 < \sigma^2+\| \mathbf{h}\|^2 P
\end{align}
to have non-zero computation rate \cite[Lemma 1]{nazer2011compute}. 

The original CF \cite{nazer2011compute} considers a power-constrained relay with $\bmod\ \Lambda_s$ decoding. This restricts the uncoded message and decoding operations to a prime-size finite field $\mathbb{F}_q$, which loses flexibility on lattice code design. In \cite{ordentlich2011practical} and \cite{mejri2013practical}, authors investigated an alternative CF relaying strategy, called incomplete-CF (ICF) in \cite{mejri2013practical}, where relay is power unconstrained with the desired linear combination being $\sum_{i= 1}^L a_{j, i} \mathbf{x}_i$. The linear combination is estimated as follows without $\bmod\ \Lambda_s$ in \eqref{equ_def_latticeequation}:
\begin{align} \label{equ_def_latticeequation2}
    \hat{\mathbf{y}}_j= DEC_{\Lambda_c}(\alpha_j \mathbf{y}_j).
\end{align}
Under ICF, the destination only requires $[\mathbf{a}_1, \mathbf{a}_2 \cdots, \mathbf{a}_J]$ to have rank $L$ in the real number field. Since the operation is over the real numbers, the ICF scheme allows uncoded messages $\mathbf{b} \in \mathbb{Z}^N$ and gives more flexibility on lattice code design to match system requirements for error performance and complexity.

In this paper, a 2-hop network using ICF is considered. Since the purpose of this paper is to study lattice construction design with physical layer error detection and retry decoding at the relay node, we focus on user-relay links which suffer from Rayleigh fading. The equation error rate (EER) at the relay, that is $\hat{\mathbf{y}}_j \neq \sum_{i= 1}^L a_{i, j} \mathbf{x}_i$, is measured to evaluate the proposed lattice construction and retry decoding scheme of this paper.

\section{Retry decoding for single user transmission} \label{sec_su}

In this section, we consider retry decoding for single user transmission through the AWGN channel. First, the retry decoding scheme is described with an algorithm that recursively finds a finite-length $\alpha$ candidate list in the order of probability of correct decoding. A lower bound on error rate for retry decoding is then derived by extending the list to all $\alpha \in \mathbb{R}$. The proposed candidate search algorithm and lower bound on error rate are derived for a genie-aided decoder, where the user's message is assumed to be known by the genie to detect errors. Numerical results using $E_8$ and $BW_{16}$ lattice codes verify that retry decoding achieves lower WER than one-shot decoding only using $\alpha_{MMSE}$.

\subsection{Decoding scheme} \label{sec_SU_dec_scheme}
{{}For single user transmission using finite dimensional lattices and lattice codes with one-shot decoding, the receiver applies a MMSE scaling factor $\alpha_{MMSE}$ in \eqref{equ_alpha_MMSE} to received message for lattice decoding in \eqref{equ_lattice_dec}. For unconstrained lattices, the scaling factor can be seen as $\alpha_{MMSE}= P/(P+ \sigma^2) \rightarrow 1$ since the average power $P \rightarrow \infty$. If a decoding error is detected, the proposed retry decoding scheme allows lattice decoder to adjust the scaling factor $\alpha$ to achieve a lower error rate.
}

At the receiver, the decoder stores a list of $\alpha$ candidates, called the search space. For retry decoding having $k$ levels, the list is grouped into $k$ non-overlapping subsets $\mathcal{A}_1, \mathcal{A}_2, \cdots, \mathcal{A}_k$, where each level may contain multiple $\alpha$ candidates. The decoding starts from $\alpha \in \mathcal{A}_1$; then tests $\mathcal{A}_2, \cdots, \mathcal{A}_k$ sequentially. Error detection is performed after each decoding attempt, rather than the whole level. The decoding terminates as soon as no error is detected or all candidates are tested. If all $\alpha$'s failed on error detection, the decoder may output a decoding failure to request a re-transmission.

Below, we give an algorithm to find $\alpha$ candidates which also shows how we group the list and number of element in each subset $\mathcal{A}$ using genie-aided decoding and Monte-Carlo-based search. To perform efficient decoding, $\mathcal{A}_1, \mathcal{A}_2, \cdots, \mathcal{A}_k$ are generated in the order to maximize the probability of correct decoding given the previous candidates failed. 
In this algorithm, $\mathcal{A}_i$ for decoding level $i$ contains $2^ {i- 1}$ of $\alpha$ values from which similar error performance is observed. The order within $\mathcal{A}_i$ can be arbitrary since it does not affect error performance.
For a given lattice code $\mathcal{C}$, the candidate list only depends on SNR. Therefore, the algorithm can be performed offline to generate a look-up table of decoding coefficients for all required SNR values in advance.

Details of the algorithm are as follows. Let $\mathbf{x} \in \mathcal{C}$ be a randomly generated lattice codeword and $\mathbf{z} \sim \mathcal{N}(0, \sigma^2 \mathbf{I})$. The probability of correctly decoding the received message $\mathbf{y}= \mathbf{x}+ \mathbf{z}$ is a function of $\alpha$:
\begin{align} \label{equ_prob_correct_dec}
    P(\alpha)= \Pr(DEC_{\Lambda}(\alpha \mathbf{y})= \mathbf{x}).
\end{align}
Given a search space $[\alpha_{min}, \alpha_{max}]$, the first step of the algorithm finds $\alpha_{1, 1}$ that maximizes
\begin{align}
    \alpha_{1, 1}= \mathop{\arg\max}_{\alpha_{min} \leq \alpha \leq \alpha_{max}} P(\alpha).
\end{align}
The result of the first step is given as $\mathcal{A}_1= \{\alpha_{1, 1}\}$. Empirically, we have $\alpha_{1, 1} \approx \alpha_{MMSE}$ \cite{erez2004achieving}, \cite{ferdinand2014mmse}, which could alternatively be used with almost no loss of error performance. 

For the $(k+1)$-th ($k \geq 1$) step, let $\mathcal{A}_1, \mathcal{A}_2 \cdots, \mathcal{A}_{k}$ be the candidate list found in previous steps. To find $\mathcal{A}_{k+ 1}$, we define: a) $e_k$ as the event that decoding failed for all $\alpha \in \bigcup_{i= 1}^{k} \mathcal{A}_i$; b) $\mathcal{A}'_k$ as a sorted list in ascending order, $\mathcal{A}'_k= sort(\{\alpha_{mix}, \alpha_{max}\} \cup \bigcup_{i= 1}^{k} \mathcal{A}_i)$.
The algorithm then finds $\alpha$ that maximizes the probability of correct decoding given event $e_k$ occurred as
\begin{align} \label{equ_Prob_previous_failed}
    P(\alpha | e_k)= \Pr(DEC_{\Lambda}(\alpha \mathbf{y})= \mathbf{x} |  e_k).
\end{align}
Since $P(\alpha | e_k)= 0$ for $\alpha \in \bigcup_{i= 1}^{k- 1} \mathcal{A}_i$, the search space is split into $2^{k}$ intervals bounded by adjacent elements in $\mathcal{A}'_k$. From each interval, one local optimum $\alpha$ can be found which maximizes
\begin{align}
    \alpha_{k+ 1, j}= \mathop{\arg\max}_{\alpha_{j} \leq \alpha< \alpha_{j+ 1}} P(\alpha | e_k),
\end{align}
where $\alpha_{j}, \alpha_{j+ 1} \in \mathcal{A}'_{k}$ for $j= 1, 2, \cdots, 2^{k}$. The result of the $(k+ 1)$-th step of algorithm is $\mathcal{A}_{k+ 1}= \{\alpha_{k+ 1, 1}, \alpha_{k+ 1, 2}, \cdots, \alpha_{k+ 1, 2^k}\}$. 

Fig.~\ref{fig_E8_alphasub_condi} shows $P(\alpha)$ and $P(\alpha | e_i)$ for $i= 1, 2$ using an $E_8$ lattice code, from which $\alpha$ candidates for $\mathcal{A}_1, \mathcal{A}_2$ and $\mathcal{A}_3$ are found. It is illustrated that the search space in the $(k+1)$-th step is split into $2^k$ intervals between $\alpha$'s for which $P(\alpha | e_k)= 0$. And within each subset $\mathcal{A}_{i+ 1}$, the values of $P(\alpha | e_i)$ for different candidates are close. For $i= 1$, the two maximum values of $P(\alpha | e_1)$ are approximately $0.2477$ and $0.2996$ using $\alpha_{2, 1} \approx 0.9103$ and $\alpha_{2, 2} \approx 1.0555$, respectively, indicating a fraction of $0.5473$ of messages failed decoding using $\alpha_{MMSE}$ can be corrected by retry decoding using $\mathcal{A}_2= \{\alpha_{2, 1}, \alpha_{2, 2}\}$. 

\begin{figure}[t]
    \centering
    \includegraphics[width=0.9\linewidth]{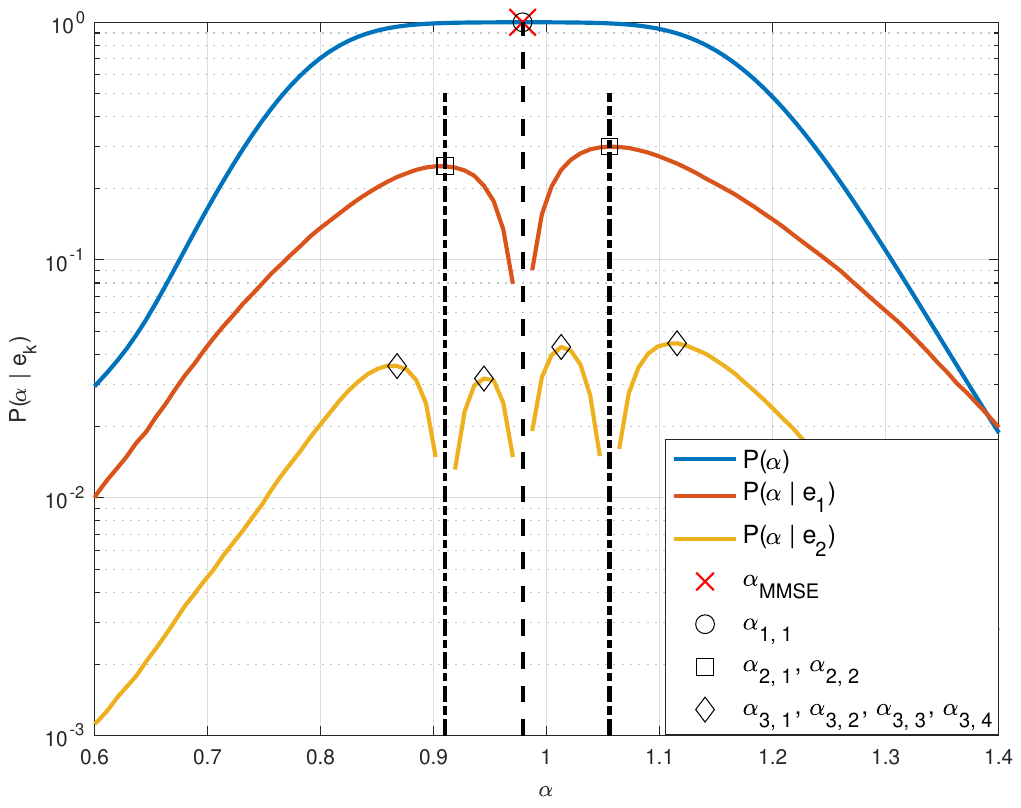}
    \caption{$P(\alpha)$, $P(\alpha | e_1)$ and $P(\alpha | e_2)$ curve for $E_8$ lattice code with hypercube shaping and code rate $R= 2$. $SNR= 17$dB so that $1- P(\alpha_{MMSE})\approx 10^{-3}$. The $\alpha_{MMSE}$ and search results $\alpha_{1, 1} \cdots \alpha_{3, 4}$ are marked at the corresponding curves.}
    \label{fig_E8_alphasub_condi}
\end{figure}

\subsection{Lower bound on error probability}  \label{sec_SU_lb}
A lower bound on error probability is derived to show the limit of retry decoding by assuming a genie-aided exhaustive search decoder for which the search space is extended from a finite-length list to all $\alpha \in \mathbb{R}$. With this assumption, the transmitted message $\mathbf{x}$ can be correctly decoded from the received message $\mathbf{y}$ if and only if $\exists \alpha \in \mathbb{R}, DEC_{\Lambda}(\alpha \mathbf{y})= \mathbf{x}$. This implies that the line connecting $\mathbf{y}$ and the origin $\mathbf{0}$ passes through the Voronoi region $\mathcal{V}(\mathbf{x})$. We refer to any such $\mathbf{y}$ as decodable and define the union of all decodable $\mathbf{y}$ as the decodable region with respect to $\mathbf{x}$.
\begin{myDef} \label{def_decodable_region}
    \rm (Decodable region) Given a non-zero lattice point $\mathbf{x}$, the decodable region is
    \begin{align} \label{equ_decodable_region}
    \mathcal{D}(\mathbf{x})= \{\frac{1}{\alpha} \mathbf{u} | \alpha \in \mathbb{R} \setminus 0, \mathbf{u} \in \mathcal{V}(\mathbf{x})\}.
\end{align}
\end{myDef}
Geometrically, $\mathcal{D}(\mathbf{x})$ forms an $N$-dimensional cone region with vertex at the origin $\mathbf{0}$. 
{{}A 2-dimensional example is illustrated in Fig.~\ref{fig_decodeable_Z2}. Suppose $\mathbf{x}_1= [5, 0]^t$ is transmitted, then $\mathbf{y}_1$ is decodable and $\mathbf{y}_1'$ is non-decodable.}
The decoding error probability of the genie-aided exhaustive search decoder for a given $\mathbf{x}$ is obtained as 
\begin{align}
    P_{e, Dec}= 1-  \Pr(\mathbf{y} \in \mathcal{D}(\mathbf{x})).
\end{align}
{{} However, the area of $\mathcal{D}(\mathbf{x})$ depends on the value of $\mathbf{x}$ (not only the underlying lattice $\Lambda$ and the message power $\|\mathbf{x}\|^2$). For example, as shown in Fig.~\ref{fig_decodeable_Z2}, $\mathbf{x}_1= [5, 0]^T$ and $\mathbf{x}_2= [4, 3]^T$, the area of $\mathcal{D}(\mathbf{x}_1)$ and $\mathcal{D}(\mathbf{x}_2)$ are different since the angle $\theta_1 \neq \theta_2$. Therefore $P_{e, Dec}$ is hard to find in general. }
Instead, using the covering sphere $\mathcal{S}_c(\mathbf{x})$, we define 
\begin{align}
    \mathcal{D}_c(\mathbf{x})= \{\frac{1}{\alpha} \mathbf{u} | \alpha \in \mathbb{R}  \setminus 0, \mathbf{u} \in \mathcal{S}_c(\mathbf{x})\}.
\end{align}
For finite dimensional lattices, we have $\mathcal{V}(\mathbf{x}) \subset \mathcal{S}_c(\mathbf{x})$, by which $\mathcal{D}(\mathbf{x}) \subset \mathcal{D}_c(\mathbf{x})$ and $P_{e, Dec}$ is strictly lower bounded by
\begin{align} \label{equ_lowerbound_form}
    P_{e, Dec}> 1- \Pr(\mathbf{y} \in \mathcal{D}_c(\mathbf{x})).
\end{align}
Due to the circular symmetry of $N$-dimensional Gaussian noise, $\Pr(\mathbf{y} \in \mathcal{D}_c(\mathbf{x}))$ only depends on the message power $\|\mathbf{x}\|^2$ and covering radius $r_c$. {{}Similarly, using the effective sphere $\mathcal{S}_e(\mathbf{x})$, we define 
\begin{align}
    \mathcal{D}_e(\mathbf{x})= \{\frac{1}{\alpha} \mathbf{u} | \alpha \in \mathbb{R}  \setminus 0, \mathbf{u} \in \mathcal{S}_e(\mathbf{x})\},
\end{align}
from which an \emph{effective sphere estimate} of error probability is given by:
\begin{align} \label{equ_estimate_form}
    P_{e,Dec} \approx 1- \Pr(\mathbf{y} \in \mathcal{D}_e(\mathbf{x})).
\end{align}
An analytical form of the lower bound on error probability in \eqref{equ_lowerbound_form} and the effective sphere estimate in \eqref{equ_estimate_form} are derived for AWGN transmission as follows.

}

\begin{figure}[t]
    \centering
    \includegraphics[width=0.9\linewidth]{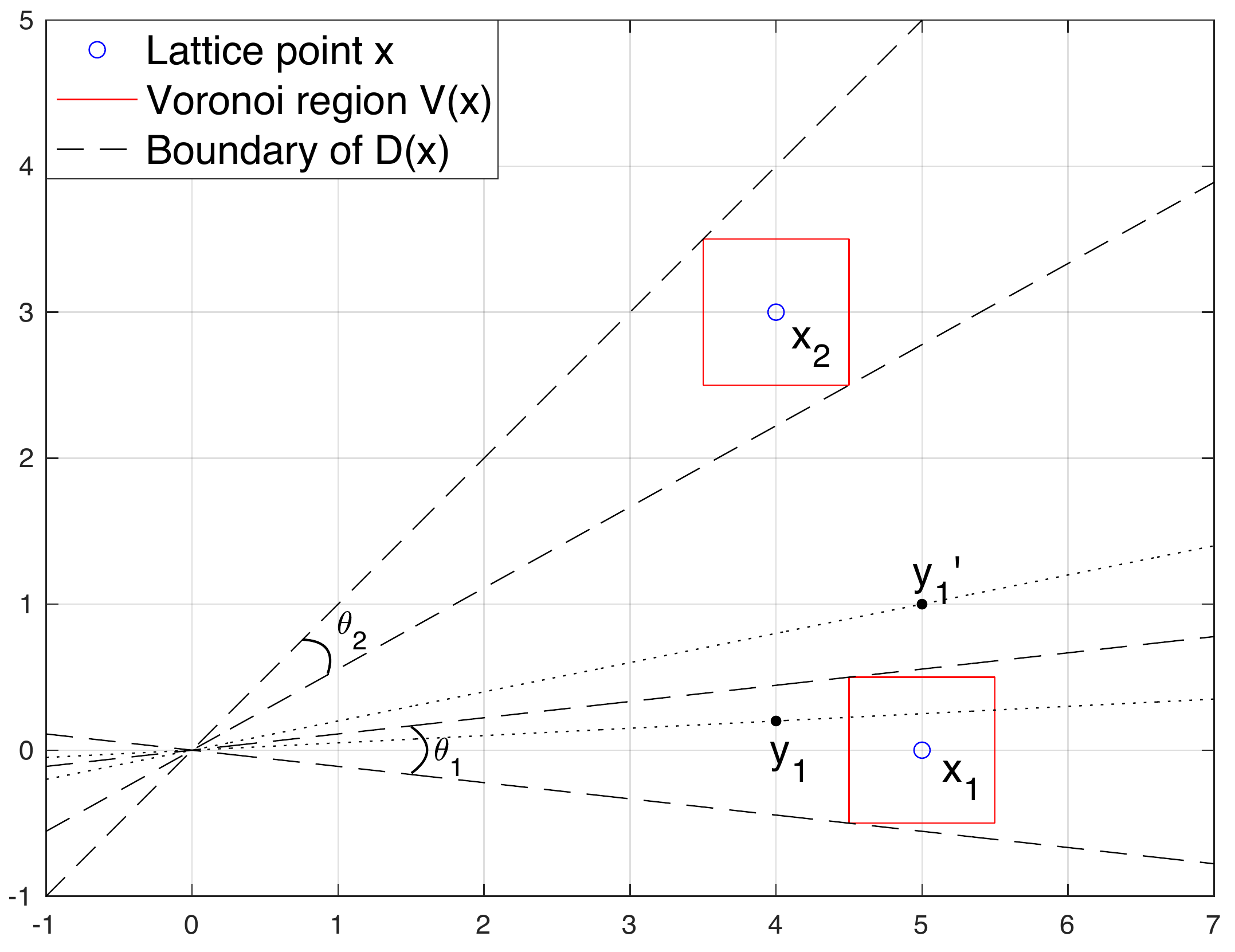}
    \caption{$\mathcal{D}(\mathbf{x})$ of $Z_2$ lattice with $\mathbf{x}_1= [5, 0]^T$ and $\mathbf{x}_2= [4, 3]^T$. With respect to $\mathbf{x}_1$, a decodable $\mathbf{y}_1$ and a non-decodable $\mathbf{y}_1'$ are plotted. The area of $\mathcal{D}(\mathbf{x}_1)$ and $\mathcal{D}(\mathbf{x}_2)$ are different since the angle $\theta_1 \neq \theta_2$.}
    \label{fig_decodeable_Z2}
\end{figure}

\begin{myTheor} \label{theo_su_bound}
\rm  Let non-zero $\mathbf x$ be a lattice point of an $N \geq 2$ dimensional lattice $\Lambda$ having covering radius $r_c$. Let $P_{\mathbf{x}}= \|\mathbf{x}\|^2/N$ and $\sigma^2$ be per-dimensional message and noise power. With the restriction $r_c^2 < N P_{\mathbf x}$, the word error probability for retry decoding is lower bounded by:
\begin{align} \label{equ_sphere_bound}
    P_{e,Dec} > 1- \int_{-\infty}^{\infty} \frac{1}{\sqrt{2 \pi \sigma^ 2}} e^{- \frac{z^2}{2 \sigma^2}} (1 - h(z)) dz
\end{align}
where if $N$ is odd:
\begin{align*}
    h(z) =  e^{- t}\left(\sum_{k= 0}^{(N- 3)/ 2} \frac{t^k}{k!}\right)
\end{align*}
and if $N$ is even:
\begin{align*}
    h(z) =  \erfc(t^ {1/2}) + e^{- t}\left(\sum_{k= 1}^{(N- 2)/ 2} \frac{t^{k- 1/2}}{(k- 1/ 2)!}\right)
\end{align*}
with $t= f^2(z) / (2 \sigma^2)$ and $f(z)= \left| \frac{r_c}{\sqrt{N P_{\mathbf x} - r_c^ 2}} z + \sqrt{\frac{N P_{\mathbf x} r_c^ 2}{N P_{\mathbf x} - r_c^ 2}} \right|$. 
\end{myTheor}

\begin{proof}
Here we give a sketch of the proof. The details are provided in \cite{xue2022lower}. For the AWGN transmission, the lower bound is obtained by taking the integral of Gaussian noise $\mathbf{z}$ over $\mathcal{D}_c(\mathbf{x})- \mathbf{x}$ as:
\begin{align}
    P_{e, Dec} & > 1- \Pr(\mathbf{z} \in \mathcal{D}_c(\mathbf{x})- \mathbf{x}) \nonumber \\
                             & = 1- \int_{\mathcal{D}_c(\mathbf{x})- \mathbf{x}} g_N(\mathbf{z})\, d\mathbf{z} \label{equ_int_ndim_noise}
\end{align}
where $g_N(\mathbf{z})$ is the density function of $\mathbf{z}$. Due to the circular symmetry of Gaussian noise, a rotated orthogonal coordinate system centered at $\mathcal{D}_c(\mathbf{x})- \mathbf{x}$ is built where $z_1$ axis connects the vertex of $\mathcal{D}_c(\mathbf{x})$ and $\mathbf{x}$ (See Fig.~\ref{fig_possible_noise_region} for an example in 2 dimension). Since components in $\mathbf{z}$ are i.i.d, we have:
\begin{align}
    \int_{\mathcal{D}_c(\mathbf{x})- \mathbf{x}} g_N(\mathbf{z})\, d\mathbf{z} = \int_{-\infty}^{\infty} \frac{1}{\sqrt{2 \pi \sigma^ 2}} e^{- \frac{z_1^2}{2 \sigma^2}} P_s(r_z)  d z_1,
\end{align}
where $P_s(r_z)$ is integral of Gaussian noise over an $(N-1)$-dimensional sphere with radius $r_z$, which is a function of $z_1$ and has a closed form \cite[Theorem 2.2]{tarokh1999universal}. Then the lower bound is obtained as in \eqref{equ_sphere_bound} where only a single integral is used.
\end{proof}

{{}
\begin{myCorollary} \label{corollary_effec_sphere_esti}
    \rm The effective sphere estimate in \eqref{equ_estimate_form} can be obtained using the same form as \eqref{equ_sphere_bound} by replacing the covering radius $r_c$ with the effective radius $r_e$.
\end{myCorollary}
}

\begin{figure}[t]
    \centering
    \includegraphics[width=0.9\linewidth]{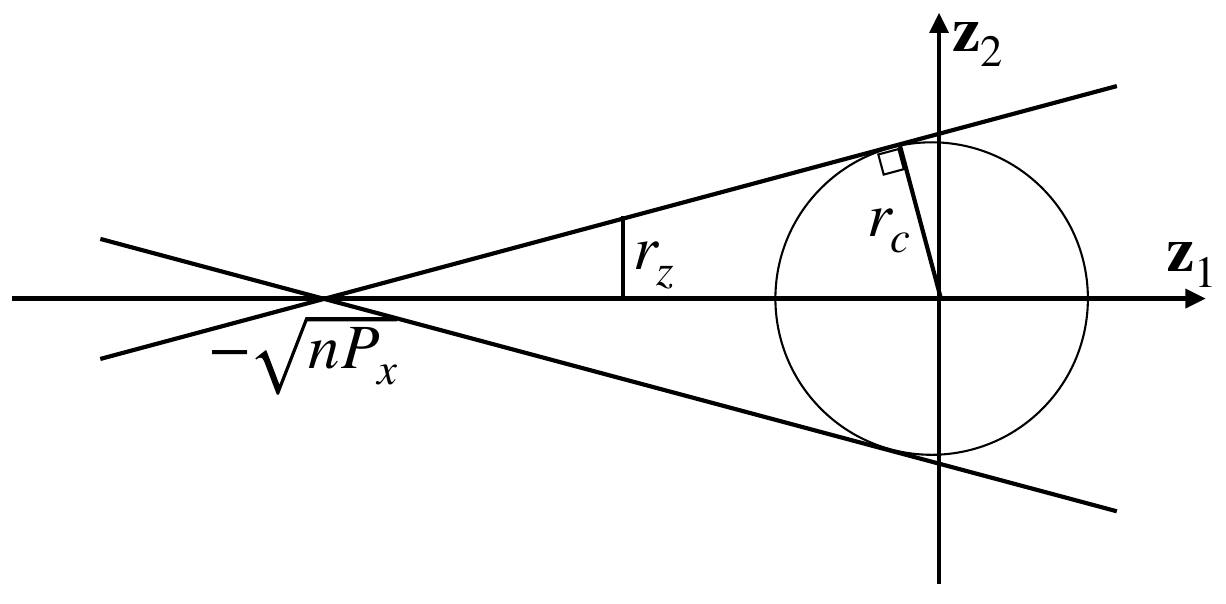}
    \caption{Example of rotated $\mathcal{D}_c(\mathbf{x})- \mathbf{x}$ in 2 dimensions. The vertex of the decodable region is $(-\sqrt{n P_{\mathbf x}},0)$, where $P_{\mathbf x}$ is the message power. $r_c$ is the covering radius.}
    \label{fig_possible_noise_region}
\end{figure}    

Theorem~\ref{theo_su_bound} is given with respect to a lattice $\Lambda$, but the result can be extended to lattice codes $\mathcal{C}= \Lambda_c/ \Lambda_s$ under lattice decoding by letting $P_{\mathbf{x}}$ be the average per-dimensional power of $\mathcal{C}$. The lower bound and effective sphere estimate are particularly suitable for analyzing error probabilities of low dimensional lattices, which have known covering radius and well-studied geometric properties. 
{{} Since the covering sphere only satisfies $\mathcal{V}(\mathbf{x}) \subset \mathcal{S}_c(\mathbf{x})$, the lower bound may not be tight. The tightness depends on the lattice covering thickness defined as $\Theta(\Lambda)= V(\mathcal{S}_c)/ V(\Lambda)$, where it may loose when lattice has large covering thickness $\Theta(\Lambda)$. For finite $N$, an upper and a lower bound on the thinnest covering are given in \cite[Chapter 2]{conway1993sphere} as
\begin{align}
    \frac{N}{e \sqrt{e}} \lessapprox \Theta \leq N \ln N+ N \ln \ln N+ 5N,
\end{align}
where the lower bound, implying the best achievable covering thickness, increases linearly as $N$ grows. However, how the covering thickness and the tightness of the lower bound are related is still an open question. The accuracy of the effective sphere estimate depends on the normalized second moment (NSM) of $\mathcal{V}(\mathbf{x})$, which has a lower bound obtained from an $N$-dimensional sphere as $\frac{\Gamma(N/ 2+ 1)^{2/ N}}{\pi (N+ 2)}$ \cite{zamir2014lattice}. A good lattice quantizer implies its Voronoi region is close to a sphere. For such lattices, the effective sphere estimate gives a relatively accurate estimate of WER of the genie-aided exhaustive search decoding.

}

Fig.~\ref{fig_E8_BW16_example_SU} shows the evaluation of the lower bound in \eqref{equ_lowerbound_form} and the effective sphere estimate in \eqref{equ_estimate_form} using $E_8$ and $BW_{16}$ lattice codes with hypercube shaping. For comparison, genie-aided exhaustive search decoding is evaluated to indicate the best retry decoding can achieve. A large search space of $[0.15, 0.85]$ for $E_8$ and $[0.4, 1.6]$ for $BW_{16}$ lattice codes with 200 $\alpha$ candidates allocated uniformly are assumed, for which $0.5$ dB and $0.4$ dB gains achieved at WER$= 10^{-5}$ for $E_8$ and $BW_{16}$ lattice codes, respectively. Since $E_8$ and $BW_{16}$ lattices are the best known quantizer among 8 and 16 dimensional lattices\cite{conway1993sphere} respectively, this implies a sphere-like Voronoi region from which \eqref{equ_estimate_form} gives a relatively accurate estimate of WER of the genie-aided exhaustive search decoding. The WER performance with a finite number of decoding attempts using $\mathcal{A}_1$ and $\mathcal{A}_2$ are also evaluated, which approaches that of exhaustive search decoding. This indicates that, by appropriately selecting $\alpha$ candidates, we may need only a small number of retry attempts to approach the error performance of the exhaustive search decoding.

\begin{figure}
    \centering
    \includegraphics[width=0.9\linewidth]{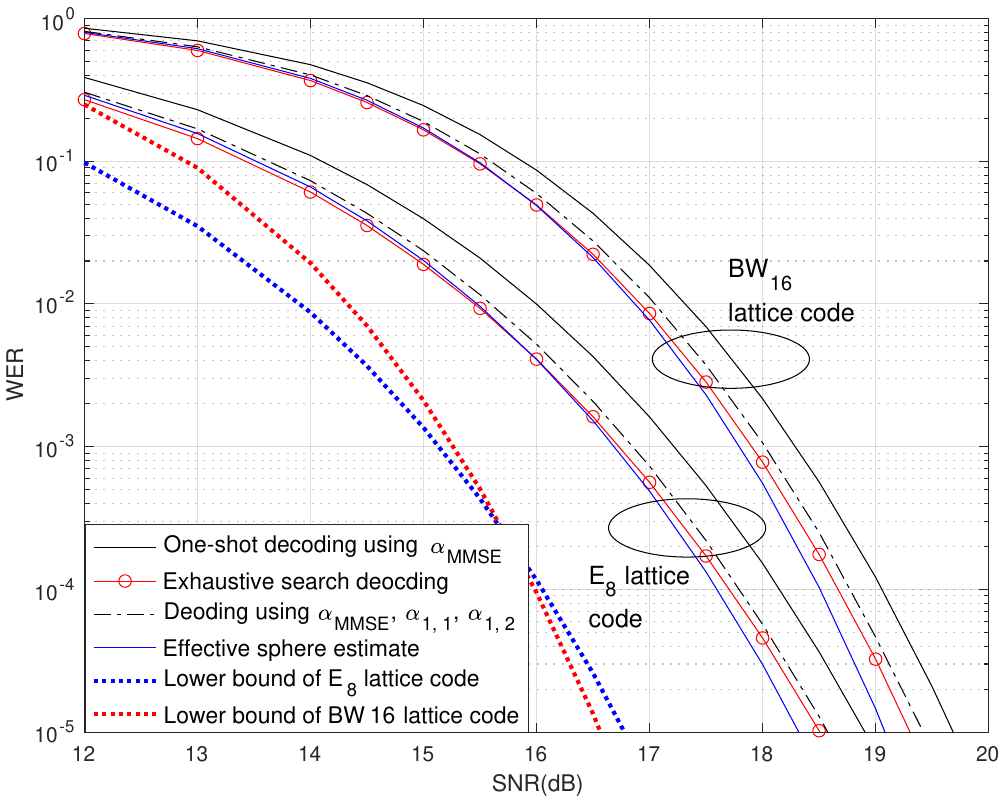}
    \caption{Numerical example using $E_8$ and $BW_{16}$ lattice codes with hypercube shaping and code rates $R_{8}= 2$ and $R_{16}= 2.25$, respectively.}
    \label{fig_E8_BW16_example_SU}
\end{figure}


{{}
\section{Retry decoding for CF relaying} \label{sec_mac_CFrelay}
This section describes retry decoding for CF relaying using the ICF scheme \cite{mejri2013practical}, where the linear combination is estimated using \eqref{equ_def_latticeequation2} without the $\bmod\ \Lambda_s$ operation. The retry decoding scheme and the decoding coefficient selection strategy are given. 

Recall that CF relaying uses a coefficient set $\{\mathbf{a}, \alpha\}$ to estimate a linear combination of users' messages. The received message in \eqref{equ_CF_j_relay_rece} is scaled by $\alpha$ at decoder as (omitting relay index $j$ here)
\begin{align} \label{equ_CF_message}
    \alpha \mathbf{y}= \sum_{i= 1}^L a_{i} \mathbf{x}_i+ \sum_{i= 1}^L (\alpha h_{i}- a_{i}) \mathbf{x}_i+ \alpha \mathbf{z},
\end{align}
from which the equivalent noise is defined as
\begin{align} \label{equ_CF_message_equi_noise}
    \Tilde{\mathbf{z}}= \sum_{i= 1}^L (\alpha h_{i}- a_{i}) \mathbf{x}_i+ \alpha \mathbf{z}.
\end{align}
A decoding error happens if $\Tilde{\mathbf{z}} \not\in \mathcal{V}(\mathbf{0})$. Retry decoding for CF relaying aims to change $\mathbf{a}$ and/or $\alpha$ at the decoder so that the new equivalent noise $\Tilde{\mathbf{z}}' \in \mathcal{V}(\mathbf{0})$ to reduce error rate. 

\subsection{Decoding schemes}   \label{sec_mac_CFrelay_dec_scheme}
Assume the current coefficient set $\{\mathbf{a}, \alpha\}$ failed decoding. Now we have freedom on changing $\mathbf{a}$ and/or $\alpha$ for retry decoding. The following three possible schemes are considered:
\begin{enumerate} [1)]
    \item select a new integer coefficient $\mathbf{a}'$, compute its MMSE $\alpha'$ using \eqref{equ_def_optalpha}; 
    \item keep integer coefficient $\mathbf{a}$ unchanged, select a new $\alpha'$;
    \item combine 1) and 2).
\end{enumerate}

}

For scheme 1, consider a $k$-level retry decoding using a length-$k$ candidate list $\{\mathbf{a}_1, \alpha_1\}, \{\mathbf{a}_2, \alpha_2\}, \cdots, \{\mathbf{a}_{k}, \alpha_{k}\}$ with computation rate $R_{c, 1} \geq R_{c, 2} \geq \cdots \geq R_{c, k}$. For $j= 1, \cdots, k$, the scaling factor $\alpha_j$ is obtained by \eqref{equ_def_optalpha} using $\mathbf{a}_j$. The decoding starts from $\{\mathbf{a}_1, \alpha_1\}$ and tests $\{\mathbf{a}_2, \alpha_2\}, \cdots, \{\mathbf{a}_k, \alpha_k\}$ sequentially with error detection performed after each attempt. At the $j$-th attempt using $\{\mathbf{a}_{j}, \alpha_{j}\}$, the estimated linear combination $\hat{\mathbf{x}}_j= DEC_{\Lambda_c}(\alpha_{j} \mathbf{y})$ is obtained by the lattice decoder for $\Lambda_c$. The decoding terminates as soon as $\hat{\mathbf{x}}_j$ passes error detection or all candidates are tested. If all candidates fail, the decoder may output a decoding failure and trigger a re-transmission request without forwarding the error-containing message into the system. 

Next we show how to generate a coefficient list for $\{\mathbf{a}, \alpha\}$. For a given channel $\mathbf{h}$, first initialize a search space $\mathcal{S}_a$ for $\mathbf{a}$ according to \eqref{equ_CF_a_condi}. A size reduction of $\mathcal{S}_a$ is then performed before searching.
{{}
\begin{myLemma} \label{lemma_CF_aiaj}
    \rm For any $\mathbf{a}_i \in \mathcal{S}_a$, all $\mathbf{a}_j$'s that satisfy $\mathbf{a}_j = m \mathbf{a}_i$, with integer $m = -1$ or $|m|> 1$, does not improve error performance and are eliminated from $\mathcal{S}_a$.
\end{myLemma}
\begin{proof}
    Recall that $\alpha$ is obtained by \eqref{equ_def_optalpha}. For $m= -1$, it is trivial that $\mathbf{a}_j= -\mathbf{a}_i$ and $\alpha_j= -\alpha_i$ resulting in equal error performance. For $|m|> 1$, we have $\mathbf{a}_j= m\mathbf{a}_i$ and $\alpha_j= m\alpha_i$. Within a same transmission, messages $\mathbf{x}_i$, for $i= 1, 2, \cdots, L$, are unchanged. The equivalent noise for $\{\mathbf{a}_i, \alpha_i\}$ and $\{\mathbf{a}_j, \alpha_j\}$ satisfy $\Tilde{\mathbf{z}}_j= m \Tilde{\mathbf{z}}_i > \Tilde{\mathbf{z}}_i$, while the direction of noise vector does not change. If $\Tilde{\mathbf{z}}_i \not\in \mathcal{V}(\mathbf{0})$, then $\Tilde{\mathbf{z}}_j \not\in \mathcal{V}(\mathbf{0})$ is also satisfied. Therefore, we have $\Pr(\Tilde{\mathbf{z}}_j \not\in \mathcal{V}(\mathbf{0})) \geq \Pr(\Tilde{\mathbf{z}}_i \not\in \mathcal{V}(\mathbf{0}))$, that is the decoding error probability using $\{\mathbf{a}_j, \alpha_j\}$ is never smaller than that using $\{\mathbf{a}_i, \alpha_i\}$.
\end{proof}
\noindent Using the size reduced search space $\mathcal{S}_a$, the integer coefficients $\mathbf{a}_1, \mathbf{a}_2, \cdots, \mathbf{a}_k$ are selected with the $k$ greatest computation rates and the corresponding $\alpha_1, \alpha_2, \cdots, \alpha_k$ are computed using \eqref{equ_def_optalpha}.

Scheme 2 is similar to the retry decoding scheme for single user transmission. The $\alpha$ search algorithm discussed in Section~\ref{sec_SU_dec_scheme} can be applied here. For a $k$-level retry decoding, the candidate list is given as $\{\mathbf{a}, \mathcal{A}_{1}, \cdots, \mathcal{A}_{k}\}$. Note that Lemma~\ref{lemma_CF_aiaj} does not apply to scheme 2, since the direction of the equivalent noise vector changes for different $\alpha$'s. For example, suppose $\alpha'= m \alpha$, the equivalent noise of $\alpha'$ is
\begin{align} \label{equ_CF_message_equi_noise_s2}
    \Tilde{\mathbf{z}}' & = \sum_{i= 1}^L (\alpha' h_{i}- a_{i}) \mathbf{x}_i+ \alpha' \mathbf{z} \nonumber \\
    & = \sum_{i= 1}^L (m \alpha h_{i}- a_{i}) \mathbf{x}_i+ m \alpha \mathbf{z} \nonumber \\
    & = m \Tilde{\mathbf{z}}+ (m- 1)\sum_{i= 1}^L a_{i} \mathbf{x}_i,
\end{align}
where $\Tilde{\mathbf{z}}'$ is not a scaled version of $\Tilde{\mathbf{z}}$. The two terms in \eqref{equ_CF_message_equi_noise_s2} may have different directions so that vector cancellation may happen resulting a smaller equivalent noise even for $|m|> 1$. 
However, the optimal $\alpha$ obtained in \eqref{equ_def_optalpha} depends on the channel coefficient $\mathbf{h}$. The $\alpha$ candidate list should be generated differently for each $\mathbf{h}$, which might be impractical if channel is time variant. Therefore, we consider this scheme for fixed channel assumption, where generating $\alpha$ candidate list and retry decoding follow Section~\ref{sec_SU_dec_scheme}; while for time variant channel, we consider scheme 1) only.

Scheme 3 is a combination of scheme 1 and scheme 2. Due to the restrictions of scheme 2, this scheme is also considered for fixed channel assumption. For a $(k, m)$-level retry decoding, denote the candidate list as $\{\mathbf{a}_1, \mathcal{A}_{1, 1}, \cdots, \mathcal{A}_{1, m}\}, \cdots, \{\mathbf{a}_{k}, \mathcal{A}_{k, 1}, \cdots, \mathcal{A}_{k, m}\}$. Steps for finding the candidate list and retry decoding are straightforward. The decoder first finds a list of integer coefficients $\mathbf{a}_1, \mathbf{a}_2, \cdots, \mathbf{a}_k$ following scheme 1. For each $\mathbf{a}_i$, finding a list of $\alpha$ following scheme 2. And decoder applies scheme 1 and scheme 2 alternately for retry.

\subsection{Scheme 1 vs scheme 2}  \label{sec_mac_CFrelay_s1_vs_s2}
In the previous subsection, we gave the decoding schemes for CF relaying, where decoder has freedom to change $\mathbf{a}$ and/or $\alpha$. It is noticed that, if channel coefficients are fixed, both scheme 1 and scheme 2 are suitable; if channel coefficients are time variant, only scheme 1 is suitable. Next we give a discussion on the differences scheme 1 and scheme 2 with different, but fixed, values of $\mathbf{h}$. 

Recall the equivalent noise $\Tilde{\mathbf{z}}$ in \eqref{equ_CF_message_equi_noise}. Except the scaled Gaussian noise $\alpha \mathbf{z}$, another term $\sum_{i= 1}^L (\alpha h_{i}- a_{i}) \mathbf{x}_i$, referred to as integer approximation error, is nonnegligible in CF relaying. Since large $\alpha$ scales the Gaussian noise, it is desired to use a relatively small $\alpha$ by which $\alpha \mathbf{h}$ approximates an integer vector $\mathbf{a}$. Error correction of retry decoding is to use new coefficient set to reduce $\Tilde{\mathbf{z}}$ by reducing either $\alpha \mathbf{z}$ or $(\alpha \mathbf{h}- \mathbf{a})$.

Given a channel $\mathbf{h}$, we say a `good' approximation exists if a `small' integer approximation error can be achieved by using a `small' MMSE $\alpha$. First, suppose a good approximation exists for a given $\mathbf{h}$ with a coefficient set $\{\mathbf{a}, \alpha\}$. Retry using a different $\mathbf{a}'$ may significantly increase the integer approximation error; and having no guarantee on reducing the scaled Gaussian noise to overcome the extra integer approximation error. On the other hand, the decoder must select a coefficient set $\{\mathbf{a}, \alpha\}$ with small integer approximation error but large $\alpha$; or large integer approximation error but small $\alpha$, where either contributes to the equivalent noise $\Tilde{\mathbf{z}}$ in a different way. In this case, changing $\mathbf{a}$ and $\alpha$ as a set explores more possibilities of $\Tilde{\mathbf{z}}$ to correct more errors. It is noticed that for scheme 1, changing coefficient set $\{\mathbf{a}, \alpha\}$ may significantly change the equivalent noise, which can be seen as a relatively aggressive strategy. While, scheme 2 is relatively conservative since the equivalent noise can be changed continuously if $\alpha$ is changed continuously over the real number space. 
Intuitively, we expect that scheme 2 is more suitable for retry decoding than scheme 1 if a good approximation exists for a given $\mathbf{h}$; otherwise scheme 1 is more suitable. The following numerical example justifies the discussion above.

\begin{myExp} \label{exp_CF_fixh_diff_dec_scheme}
    \rm Consider two-user case with normalized channel gain $\|\mathbf{h}\|= 1$. Let $\mathbf{h}_1= [0.6095, 0.7928]^T$ and $\mathbf{h}_2= [0.4299, 0.9029]^T$, where $\mathbf{h}_1$ has a poor approximation with $h_{1, 1}/ h_{1, 2} \approx 1/1.3$ and $\mathbf{h}_2$ has a good approximation with $h_{2, 1}/ h_{2, 2} \approx 1/2.1$. Suppose SNR=30dB and Gaussian noise variance $\sigma^2= 1$. The optimal and second-best coefficient set for $\mathbf{h}_1$ are $\{[3, 4]^T, 4.9946\}$ and $\{[1, 1]^T, 1.4009\}$; for $\mathbf{h}_2$ are $\{[1, 2]^T, 2.2334\}$ and $\{[2, 5]^T, 5.3688\}$, respectively. Compute the equivalent noise variance using $N_{e}= \alpha^2 \sigma^2+ P \|\alpha \mathbf{h}- \mathbf{a}\|^2$ \cite{nazer2011compute}, where $P$ is given in \eqref{equ_average_power}. For $\mathbf{h}_1$, $N_{e, 1, opt} \approx 28.5230$ and $N_{e, 1, sec} \approx 35.5625$; for $\mathbf{h}_2$, $N_{e, 2, opt} \approx 6.8416$ and $N_{e, 2, sec} \approx 147.1523$. 
    When scheme 1 is applied, or $\mathbf{h}_1$, due to the poor integer approximation, the optimal coefficient set gives larger $N_{e}$ than $\mathbf{h}_2$. However, the increase of $N_{e}$ for $\mathbf{h}_1$ is less significant when the second-best coefficient set is applied for retry. 
    Fig.~\ref{fig_BW16_fixh_diff_ratio_compare} shows the EER for $\mathbf{h}_1$ and $\mathbf{h}_2$ using $BW_{16}$ lattice code as used in Fig.~\ref{fig_E8_BW16_example_SU}. It is observed that for $\mathbf{h}_1$, scheme 1 achieves larger gain than scheme 2; for $\mathbf{h}_2$, scheme 1 achieves almost no gain while scheme 1 achieves almost same gain as that of $\mathbf{h}_1$, which justifies the discussion given above. 

    \begin{figure}
        \centering
        \includegraphics[width=0.9\linewidth]{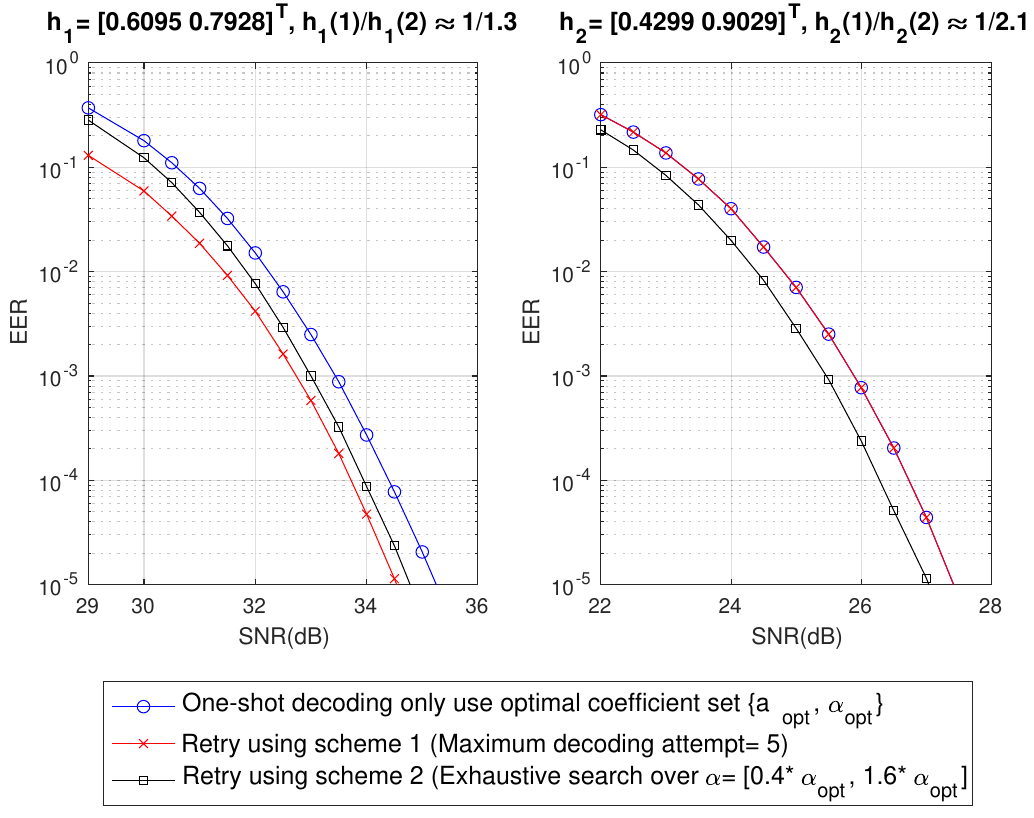}
        \caption{Comparison of EER performance for $\mathbf{h}_1= [0.6095, 0.7928]^T$ and $\mathbf{h}_2= [0.4299, 0.9029]^T$. Maximum number of decoding attempts is set to be sufficient large for each scheme. $BW_{16}$ lattice code is applied as used in Fig.~\ref{fig_E8_BW16_example_SU} and error detection is genie aided.}
        \label{fig_BW16_fixh_diff_ratio_compare}
    \end{figure}

\end{myExp}

}

\section{Lattice construction with error detection} \label{sec_code_construct}
This section gives a lattice construction which adds physical layer error detection ability to perform the retry decoding discussed in Section~\ref{sec_su} and Section~\ref{sec_mac_CFrelay}. The proposed lattice construction is given by restricting the least significant bits (LSB) of lattice uncoded messages using a binary linear block code (LBC) $\mathcal{C}_b$, referred as an LBC-embedded lattice. For practical lattice design, a CRC-embedded lattice is considered as a special case of LBC-embedded lattices, by letting $\mathcal{C}_b$ be a CRC code. For CF relaying using the ICF scheme, a condition on the shaping lattice design is given so that the relay can detect errors in the linear combinations without knowledge of the individual users' messages. Lastly, the probability of undetected error is analyzed to measure the error detection capability, along with methods to estimate it. 

\subsection{Lattice construction} \label{sec_code_construct_def}
\begin{myDef} \label{def_CRC_lattice}
    \rm (LBC-embedded lattice) For an $N$-dimensional lattice $\Lambda$ with generator matrix $\mathbf{G}$, and a binary linear block code $\mathcal{C}_b$ with block length $N$, the constellation $\Lambda'$ after embedding $\mathcal{C}_b$ into $\Lambda$ is defined as:
    \begin{align} \label{equ_def_lambda_p}
        \Lambda'= \{\mathbf{G} \mathbf{b} | \mathbf{b} \in \mathbb{Z}^N, \mathbf{b}_{LSB} \in \mathcal{C}_b\},
    \end{align}
    where the LSB vector $\mathbf{b}_{LSB}$ is obtained by the $\bmod 2$ operation as:
    \begin{align} \label{equ_LSB_extract}
        \mathbf{b}_{LSB}= \mathbf{b} \bmod 2.  
    \end{align}
\end{myDef}

In later discussion in this paper, we refer to the lattice $\Lambda$ before embedding $\mathcal{C}_b$ as the \emph{base lattice} to distinguish it from the LBC-embedded lattice $\Lambda'$. Although the definition does not explicitly show that $\Lambda'$ forms a lattice, Theorem~\ref{theo_CRC_form_lattice} states that $\Lambda'$ is indeed a lattice by showing that $\Lambda'$ forms an additive subgroup as in Definition~\ref{def_lattice}.
\begin{myTheor}
 \label{theo_CRC_form_lattice}
    \rm The constellation $\Lambda'$ defined in \eqref{equ_def_lambda_p} is a sublattice of its base lattice $\Lambda$.
\end{myTheor}

\begin{proof}
First, $\Lambda' \subseteq \Lambda$ is straightforward from \eqref{equ_def_lambda_p} as the domain of uncoded messages of $\Lambda'$ is a subset of that of $\Lambda$.

By Definition~\ref{def_lattice}, $\Lambda'$ is a lattice if $\Lambda'$ forms an additive subgroup in $\mathbb{R}^N$ that has: a) identity element; b) inverse element; c) associativity; d) commutativity; e) closure. Let $\mathbf{x} \in \Lambda'$ and the corresponding uncoded message be $\mathbf{b}= \mathbf{G}^{-1} \mathbf{x}$.
The \emph{identity element} is $\mathbf{0} \in \Lambda'$ because the LSB vector of the all-zero vector is always a codeword of $\mathcal{C}_b$. For the \emph{inverse element}, given $\mathbf{x} \in \Lambda'$, $- \mathbf{x}$ has same LSB vector as $\mathbf{x}$, by which $- \mathbf{x} \in \Lambda'$. \emph{Associativity} and \emph{commutativity} are trivial because the addition operation is over the real number space. \emph{Closure} is obtained by the linearity of $\mathcal{C}_b$. Let $\mathbf{x}_1, \mathbf{x}_2 \in \Lambda'$, $\mathbf{b}_1= \mathbf{G}^{-1} \mathbf{x}_1$ and $\mathbf{b}_2= \mathbf{G}^{-1} \mathbf{x}_2$. The LSB vector of $\mathbf{b}_1+ \mathbf{b}_2$ is:
\begin{align}
    (\mathbf{b}_1+ \mathbf{b}_2) \bmod 2 = & (\mathbf{b}_1 \bmod 2+ \mathbf{b}_2 \bmod 2) \bmod 2 \nonumber \\
    = & (\mathbf{b}_{1, LSB}+ \mathbf{b}_{2, LSB}) \bmod 2,
\end{align}
where $\mathbf{b}_{1, LSB}$ and $\mathbf{b}_{2, LSB}$ are the LSB of $\mathbf{b}_1$ and $\mathbf{b}_2$. By the linearity of $\mathcal{C}_b$, the LSB vector of $\mathbf{b}_1+ \mathbf{b}_2$ is also in $\mathcal{C}_b$. Therefore, $\mathbf{x}_1+ \mathbf{x}_2 \in \Lambda'$. This concludes the proof.
\end{proof}

Since $\Lambda'$ is a lattice, a generator matrix $\mathbf{G}'$ is found using the generator matrix of the base lattice $\Lambda$ and binary code $\mathcal{C}_b$. We first show that the uncoded message $\mathbf{b}$ of $\Lambda'$ can be considered as a construction A lattice lifted by $\mathcal{C}_b$.
\begin{myDef} \label{def_Cons_A}
    \rm (Construction A lattice \cite{zamir2014lattice}) For a $q$-ary code $\mathcal{C}_q$, a modulo-$q$ lattice is formed as:
    \begin{align}
        \Lambda_a= \mathcal{C}_q+ q \mathbb{Z}^N.
    \end{align}
\end{myDef}
It is noticed that the domain of the uncoded message $\mathbf{b}$ in \eqref{equ_def_lambda_p} can be written as $\mathcal{C}_b + 2 \mathbb{Z}^N$, which is a construction A lattice with $q= 2$. 

\begin{myProposition}  \label{prop_G_CRC_lattice}
    \rm Let a base lattice $\Lambda$ have a generator matrix $\mathbf{G}$ and binary code $\mathcal{C}_b$ have a lower triangular generator matrix $\mathbf{G}_b= \begin{bmatrix}
    \mathbf{T} \\
    \mathbf{P}
\end{bmatrix}$, where $\mathbf{T}$ is a $k \times k$ lower triangular matrix. A generator matrix of LBC-embedded lattice $\Lambda'$ is:
\begin{align} \label{equ_G_Lambda_p}
    \mathbf{G}'= \mathbf{G} \begin{bmatrix}
        \mathbf{T} & {{}\mathbf{0}_{k \times (N- k)}} \\
        \mathbf{P} & 2 \mathbf{I}_{N-k}
    \end{bmatrix},
\end{align}
where ${{}\mathbf{0}_{k \times (N- k)}}$ is $k \times (N- k)$ all-zero matrix and $\mathbf{I}_{N-k}$ is the $(N-k)$-dimensional identity matrix.
\end{myProposition}

\begin{proof}
    For $\mathcal{C}_b$ having a lower triangular $\mathbf{G}_b$, a generator matrix of construction A lattice $\Lambda_a= \mathcal{C}_b + 2 \mathbb{Z}^N$ is given as \cite{zamir2014lattice}:
    \begin{align}
        \mathbf{G}_a= \begin{bmatrix}
            \mathbf{T} & {{}\mathbf{0}_{k \times (N- k)}} \\
            \mathbf{P} & 2 \mathbf{I}_{N-k}
        \end{bmatrix}.
    \end{align}
    The uncoded message satisfying the condition in \eqref{equ_def_lambda_p} can then be expressed as $\mathbf{b}= \mathbf{G}_a \mathbf{b}'$ with $\mathbf{b}' \in \mathbb{Z}^N$, from which the lattice point of $\Lambda'$ is $\mathbf{x}= \mathbf{G} \mathbf{G}_a \mathbf{b}'$. By Definition~\ref{def_lattice}, a generator matrix of $\Lambda'$ is given as $\mathbf{G}'= \mathbf{G} \mathbf{G}_a$.
\end{proof}

\subsection{Encoding and decoding schemes} \label{sec_code_construct_endec}

\begin{figure}[t]
    \centering
    \subfloat[Encoder model]{
        \includegraphics[width=0.9\linewidth]{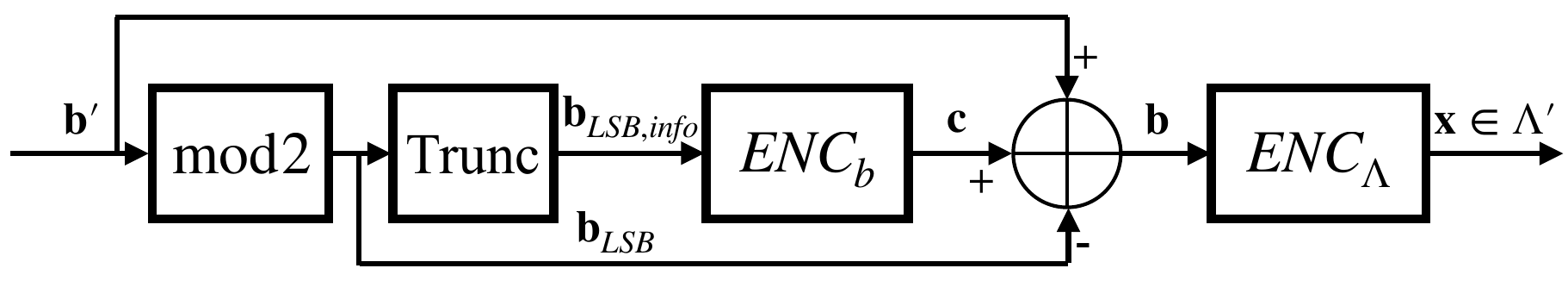}
        \label{fig_CRC_embed_enc}}
    
    \subfloat[Decoder model]{
        \includegraphics[width=0.9\linewidth]{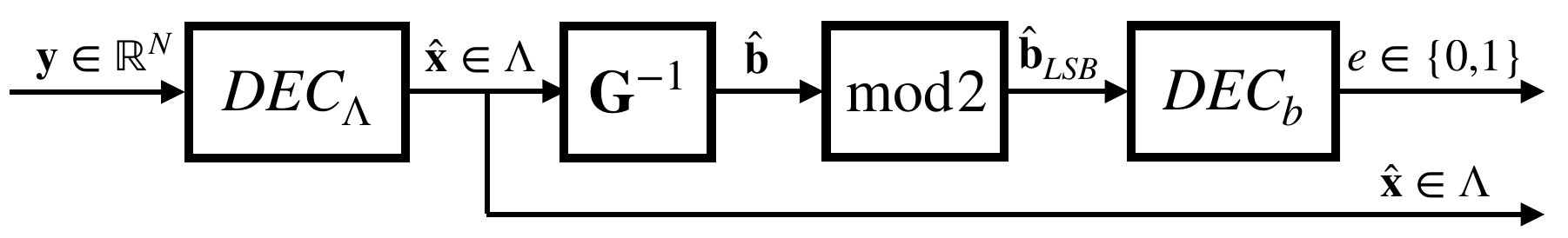}
        \label{fig_CRC_embed_dec}}
    \caption{Encoder and decoder model for LBC-embedded lattice.}
    \label{fig_CRC_embed_endec}
\end{figure}
Next encoding and decoding schemes are introduced for implementing physical layer error detection using LBC-embedded lattices. The encoder and decoder models are shown in Fig.~\ref{fig_CRC_embed_endec}. Embedding $\mathcal{C}_b$ is equivalent to removing some lattice points from the base lattice; however, the receiver uses the base lattice for decoding followed by a parity check of $\mathcal{C}_b$. $ENC_{\Lambda}$/$DEC_{\Lambda}$ and $\mathbf{G}^{-1}$ corresponds to the base lattice $\Lambda$, not the LBC-embedded lattice $\Lambda'$, and $ENC_{b}$/$DEC_{b}$ are the binary encoder/parity check of $\mathcal{C}_b$ for error detection. Denote $\mathcal{I}$ as a set of indices of $\mathbf{b}_{LSB}$ indicating the information bits for encoding $\mathcal{C}_b$. The user's message $\mathbf{b}'= [b_1', b_2', \cdots, b_N']^T$ before embedding $\mathcal{C}_b$ is defined as $b_i' \in \mathbb{Z}$ for $i \in \mathcal{I}$ and $b_i' \in 2\mathbb{Z}$ for $i \not\in \mathcal{I}$. The LSB vector $\mathbf{b}_{LSB}$ is obtained by $\bmod~2$. 
{{} The \emph{Trunc} function truncates $\mathbf{b}_{LSB}$ into $\mathbf{b}_{LSB, info}$ according to $\mathcal{I}$ to generate the input of binary encoder $ENC_b$. Then, an $\oplus$ operation combines $\mathbf{b}', \mathbf{b}_{LSB}$ and the binary codeword $\mathbf{c} \in \mathcal{C}_b$ to compute the lattice uncoded message as $\mathbf{b}= \mathbf{b}'+ \mathbf{c}- \mathbf{b}_{LSB}$, where the addition and subtraction are in real number space.} At the decoder side, since the lattice decoder $DEC_{\Lambda}$ uses $\Lambda$ instead of $\Lambda'$, a decoding algorithm for $\Lambda'$ does not need to be specified and the estimate is $\hat{\mathbf{x}} \in \Lambda$. The parity check $DEC_{b}$ following gives a 1-bit pass/fail output $e$ indicating if $\hat{\mathbf{b}}_{LSB} \in \mathcal{C}_b$ and $\hat{\mathbf{x}} \in \Lambda'$.

\subsection{Lattice codes using LBC-embedded lattices} \label{sec_code_construct_latticeC}
Forming lattice codes using LBC-embedded lattices is straightforward. In this paper, the binary code $\mathcal{C}_b$ is only embedded into the coding lattice, that is, with respect to a base lattice code $\mathcal{C}= \Lambda_c/ \Lambda_s$, the LBC-embedded lattice code is $\mathcal{C}'= \Lambda_c'/ \Lambda_s$. Applying the rectangular encoding described in Section~\ref{sec_nest_lattice}, the $i$-th component of uncoded message of $\mathcal{C}'$ is defined as: $b_i \in \mathcal{M}_i$ if $i \in \mathcal{I}$; $b_i \in \{b | b \in \mathcal{M}_i, b \bmod 2= 0\}$ if $i \not\in \mathcal{I}$, where $\mathcal{I}$ is set of indices of $\mathbf{b}_{LSB}$ indicating the information bits for encoding $\mathcal{C}_b$.
The encoding/decoding scheme follows Fig.~\ref{fig_CRC_embed_endec} by additionally including the shaping operation. Since lattice codes only contain a finite number of bits per message, the SNR penalty is non-negligible and defined as follows. 
\begin{myDef} \label{def_R_SNR_penalty}
    \rm (SNR penalty) Let a base lattice code $\mathcal{C}$ have code rate $R$ and $\mathcal{C}_b$ for embedding have $l$ parity bits. The SNR penalty $SNR_p$ of the LBC-embedded lattice code $\mathcal{C}'$ is:
    \begin{align} 
        SNR_p & = 10 \log_{10} \frac{R}{R'}(\rm{dB}), \label{equ_SNR_penalty}
    \end{align}
    where $R'= \frac{NR- l}{N}$ is the code rate of $\mathcal{C}'$.
\end{myDef}

An example is given to visualize the relationship with and without embedding $\mathcal{C}_b$.
\begin{myExp}
\rm Form $\Lambda'$ using $A_2$ lattice, with $\mathbf{G}= \big[\begin{smallmatrix}
  \sqrt{3}/ 2 & 0\\
  1/ 2 & 1
\end{smallmatrix}\big]$, and $\mathcal{C}_b$ being the single parity check code. Fig.~\ref{fig_A2lattice_CRC} illustrates the constellation of $\Lambda'$ and lattice code $\Lambda'/ 4 A_2$. For $\Lambda'$, a generator matrix $\mathbf{G}'= \big[\begin{smallmatrix}
  \sqrt{3}/ 2 & 0\\
  3/ 2 & 2
\end{smallmatrix}\big]$.
\begin{figure}[t]
    \centering
    \includegraphics[width=0.9\linewidth]{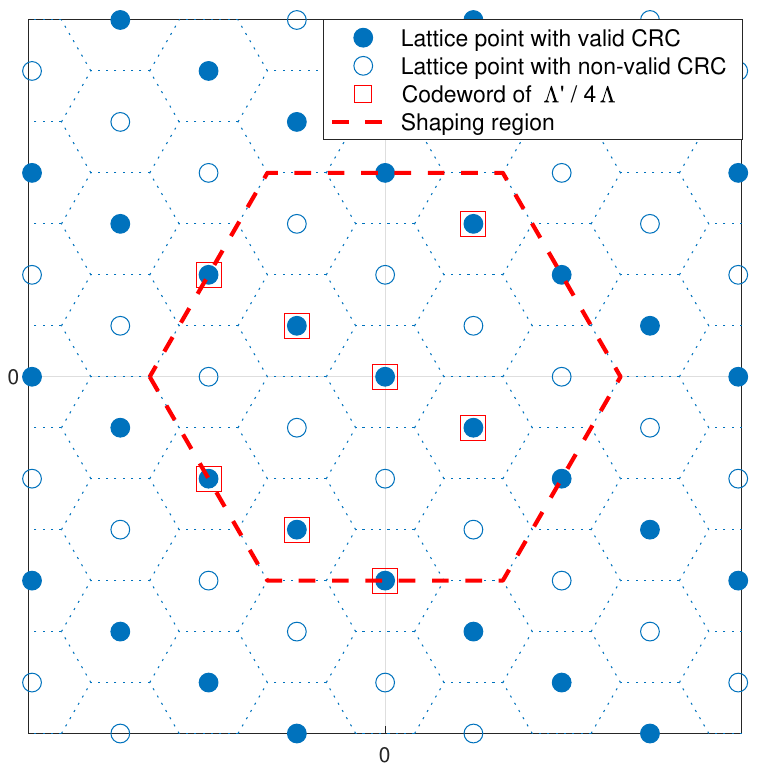}
    \caption{Constellation of A2 lattice/lattice code with single parity check code embedded.}
    \label{fig_A2lattice_CRC}
\end{figure}
\end{myExp}

\subsection{Shaping lattice design for CF relaying}
For CF relaying using ICF, a condition on designing the shaping lattice is given, by which a stand-alone relay can detect errors in linear combinations without requiring knowledge of individual users' messages. Recall that a relay performing ICF estimates a linear combination without $\bmod\ \Lambda_s$ as $\hat{\mathbf{x}}= \sum_{i= 1}^L a_i \mathbf{x}_i$, where $\hat{\mathbf{x}}$ is a member of coding lattice $\Lambda_c$ but not necessarily a member of lattice code $\mathcal{C}$. The LSB vector of $\hat{\mathbf{x}}$ is then obtained by $\hat{\mathbf{b}}_{LSB}= \mathbf{G}^{-1} \hat{\mathbf{x}} \bmod 2$. The validity of the parity check of the users' uncoded messages $\mathbf{b}_1, \mathbf{b}_2, \cdots, \mathbf{b}_L$ needs to be preserved at $\hat{\mathbf{b}}_{LSB}$ after transmission.
First, Lemma~\ref{lemma_CRC_CF} shows the validity is preserved if the transmission is uncoded, that is linearly combining uncoded messages $\mathbf{b}_i$ directly.
\begin{myLemma} \label{lemma_CRC_CF}
    \rm If $\mathbf{b}_i \bmod2 \in \mathcal{C}_b$ for $i= 1, 2, \cdots, L$, then $\sum_{i= 1}^L a_i \mathbf{b}_i \bmod 2 \in \mathcal{C}_b$ for arbitrary integers $a_1,a_2,\cdots,a_L$.
\end{myLemma}
This was shown for $L= 2$ and $a_1= a_2= 1$ when proving Theorem~\ref{theo_CRC_form_lattice}. The generalization to $L> 2$ and arbitrary integers $a_1,a_2,\cdots,a_L$ is straightforward. 

Then, we consider the lattice coded case. Let $\mathcal{C}= \Lambda_c/ \Lambda_s$, be the base lattice code before embedding binary code $\mathcal{C}_b$. Let $\mathbf{G}_c$ and $\mathbf{G}_s$ are generator matrices of $\Lambda_c$ and $\Lambda_s$, respectively. By Lemma~\ref{lemma_exist_M}, there exists an $\mathbf{M} \in \mathbb{Z}^{N \times N}$ by which $\mathbf{G}_s= \mathbf{G}_c \mathbf{M}$. 

\begin{myProposition} \label{prop_shaping_lattice_M}
    \rm Let uncoded messages $\mathbf{b}_i$ have $\mathbf{b}_{i, LSB} \in \mathcal{C}_b$ and corresponding codeword $\mathbf{x}_i$ for $i= 1, 2, \cdots, L$. The LSB vector $\hat{\mathbf{b}}_{LSB}= \left( \sum_{i= 1}^L a_i \mathbf{G}_c^{-1}  \mathbf{x}_i \right) \bmod 2 \in \mathcal{C}_b$ is satisfied for arbitrary $a_1, a_2, \cdots, a_L \in \mathbb{Z}$, if $\mathbf{M}$  only consists of even integers.
\end{myProposition}
\begin{proof}
    From \eqref{equ_latticeC_enc}, a lattice codeword $\mathbf{x}_i$ is encoded from $\mathbf{b}_i$ as:
    \begin{align}
        \mathbf{x}_i= \mathbf{G}_c \mathbf{b}_i- Q_{\Lambda_s}(\mathbf{G}_c \mathbf{b}_i)= \mathbf{G}_c \mathbf{b}_i- \mathbf{G}_s \mathbf{s}_i,
    \end{align}
    where $\mathbf{s}_i \in \mathbb{Z}^N$ so that $Q_{\Lambda_s}(\mathbf{G}_c \mathbf{b}_i)= \mathbf{G}_s \mathbf{s}_i$. The LSB vector $\hat{\mathbf{b}}_{LSB}$ is:
    \begin{align}
        \hat{\mathbf{b}}_{LSB} & =  \sum_{i= 1}^L a_i \mathbf{G}_c^{-1} \mathbf{x}_i  \bmod 2 \nonumber  \\
        & =  (\sum_{i= 1}^L a_i \mathbf{b}_i- a_i \mathbf{M} \mathbf{s}_i)  \bmod 2 \nonumber  \\
        & =  \left( \sum_{i= 1}^L a_i \mathbf{b}_i  \bmod 2 \right) \oplus \left( \sum_{i= 1}^L a_i \mathbf{M} \mathbf{s}_i  \bmod 2 \right),
    \end{align}
    where $\oplus$ is addition over the binary field.
    Lemma~\ref{lemma_CRC_CF} shows $\left( \sum_{i= 1}^L a_i \mathbf{b}_i \bmod 2\right) \in \mathcal{C}_b$. By linearity, $\hat{\mathbf{b}}_{LSB} \in \mathcal{C}_b$ if and only if $\left( \sum_{i= 1}^L a_i \mathbf{M} \mathbf{s}_i \bmod 2\right) \in \mathcal{C}_b$ for arbitrary $a_i \in \mathbb{Z}$ and $\mathbf{s}_i \in \mathbb{Z}^N$ for $i= 1, 2, \cdots, L$. If all elements of $\mathbf{M}$ are even integers, $\left( \sum_{i= 1}^L a_i \mathbf{M} \mathbf{s} \bmod 2\right)$ is the all-zero vector, indicating the membership of $\mathcal{C}_b$, which concludes the proof.
\end{proof}

\subsection{Probability of false positive on error detection} \label{sec_P_ud}
Error detection using finite dimensional $\mathcal{C}_b$ is imperfect and false positives may occur, that is for a transmitted $\mathbf{x} \in \Lambda'$, the event $\hat{\mathbf{x}} \in \Lambda'$ for $\hat{\mathbf{x}} \neq \mathbf{x}$. Such an event is referred to as an undetected error event which has probability $P_{ud}$. The following analysis is valid for evaluation of both LBC-embedded lattice $\Lambda'$ and lattice code $\mathcal{C}= \Lambda'/ \Lambda_s$, if the influence of codewords on codebook boundary can be ignored.

\begin{myDef} \label{def_p_ud}
    \rm (Probability of undetected error event) Given a desired decoding result $\mathbf{x} \in \Lambda'$ and $\hat{\mathbf{x}} \neq \mathbf{x}$, the probability of an undetected error event is:
    \begin{align} \label{equ_def_p_ud}
        P_{ud}= \Pr(\hat{\mathbf{x}} \in \Lambda' | \hat{\mathbf{x}} \neq \mathbf{x})= \frac{\Pr(\hat{\mathbf{x}} \in \Lambda', \hat{\mathbf{x}} \neq \mathbf{x})}{\Pr(\hat{\mathbf{x}} \neq \mathbf{x})}.
    \end{align}
\end{myDef}

Due to the linearity of lattices, $P_{ud}$ defined in \eqref{equ_def_p_ud} can be equivalently expressed as:
\begin{align} \label{equ_p_ud_2}
    P_{ud} = \frac{\sum_{\mathbf{e} \in \Lambda', \mathbf{e} \neq \mathbf{0}} p(\mathbf{e})}{\sum_{\mathbf{e} \in \Lambda, \mathbf{e} \neq \mathbf{0}} p(\mathbf{e})},
\end{align}
where $\mathbf{e}= \hat{\mathbf{x}}- \mathbf{x}$ is the error vector with probability $p(\mathbf{e})$.
Computing $p(\mathbf{e})$ exactly requires integrating the noise density function over $\mathcal{V}(\mathbf{e})$ which depends on the geometric properties of lattices and is impractical in most cases. Instead, we give two methods to estimate $P_{ud}$ using the kissing number and parity length of $\mathcal{C}_b$. 

\noindent \textbf{Method 1}: Denote $\mathcal{T}$ as the set of shortest non-zero vectors of $\Lambda$ and $|\{\cdot\}|$ as the cardinality of set. Kissing number is the number of shortest non-zero vectors, that is $|\mathcal{T}|$. Assume the error vector $\mathbf{e}$ is uniformly distributed over $\mathcal{T}$, then $P_{ud}$ can be estimated as:
\begin{align} \label{equ_p_ud_esti_kiss}
    P_{ud} \approx \frac{|\mathcal{T} \cap \Lambda'|}{|\mathcal{T}|}.
\end{align} 
\textbf{Method 2}: For $\mathcal{C}_b$ having $l$ parity bits, a fraction of $2^{-l}$ lattice points in $\Lambda$ have valid LSB vectors, which could cause an undetected error event. Then, $P_{ud}$ can be estimated as:
\begin{align} \label{equ_p_ud_esti_l}
    P_{ud} \approx 1/ 2^l.
\end{align}

Method 1 is more suitable for low dimensional lattices for which $\mathcal{T}$ is either known or can be found by some techniques such as list sphere decoding. The assumption on the distribution of $\mathbf{e}$ in Method 1 follows the truncated union bound thus is good for medium to high SNR with Gaussian noise. For higher dimensional lattices with unknown kissing number, method 2 gives a convenient estimate of $P_{ud}$. Table~\ref{table_p_ud} evaluates the value of $P_{ud}$ by Monte-Carlo simulation using \eqref{equ_def_p_ud} and the estimate by \eqref{equ_p_ud_esti_kiss} and \eqref{equ_p_ud_esti_l} using $BW_{16}$ lattice code with CRC of various lengths embedded. The CRC polynomials are selected to further minimize $P_{ud}$ in \eqref{equ_p_ud_esti_kiss} for given $l$. The numerical results show that both \eqref{equ_p_ud_esti_kiss} and \eqref{equ_p_ud_esti_l} give a good estimate of $P_{ud}$, whose value is dominated by the parity length $l$. Since the CRC polynomials are selected to minimize $P_{ud}$ in \eqref{equ_p_ud_esti_kiss} for each $l$, a slight lower $P_{ud}$ is obtained compared with the estimate using \eqref{equ_p_ud_esti_l}.

\begin{table*}[t] 
    \centering
    \caption{Evaluation of $P_{ud}$ of CRC-embedded $BW_{16}$ lattice codes and CRC length $l= 4, 5, 6, 7, 8$. All-zero codeword is assumed and SNR is set so that WER$\approx 10^{-3}$ for decoding using $\alpha_{MMSE}$.}
    \begin{tabular}{c|c|ccccc}
        \multicolumn{2}{c|}{CRC length($l$)}  & 4 & 5 & 6 & 7 & 8 \\
        \hline \hline
        \multirow{3}{*}{$P_{ud}$} & Monte-Carlo\eqref{equ_def_p_ud} & 5.619e-2 & 2.625e-2 & 1.205e-2 & 4.201e-3 & 1.386e-3 \\
        \cline{2-7}
        & Kissing number\eqref{equ_p_ud_esti_kiss} & 5.556e-2 & 2.593e-2 & 1.204e-2 & 4.167e-3 & 1.389e-3 \\
        \cline{2-7}
        & Parity length\eqref{equ_p_ud_esti_l} & 6.250e-2 & 3.125e-2 & 1.563e-2 & 7.813e-3 & 3.906e-3
    \end{tabular}
    
    \label{table_p_ud}
\end{table*}

{{}
\section{CRC length optimization} \label{sec_CRC_length_opt}
}

\begin{figure}[t]
    \centering
    \includegraphics[width=0.9\linewidth]{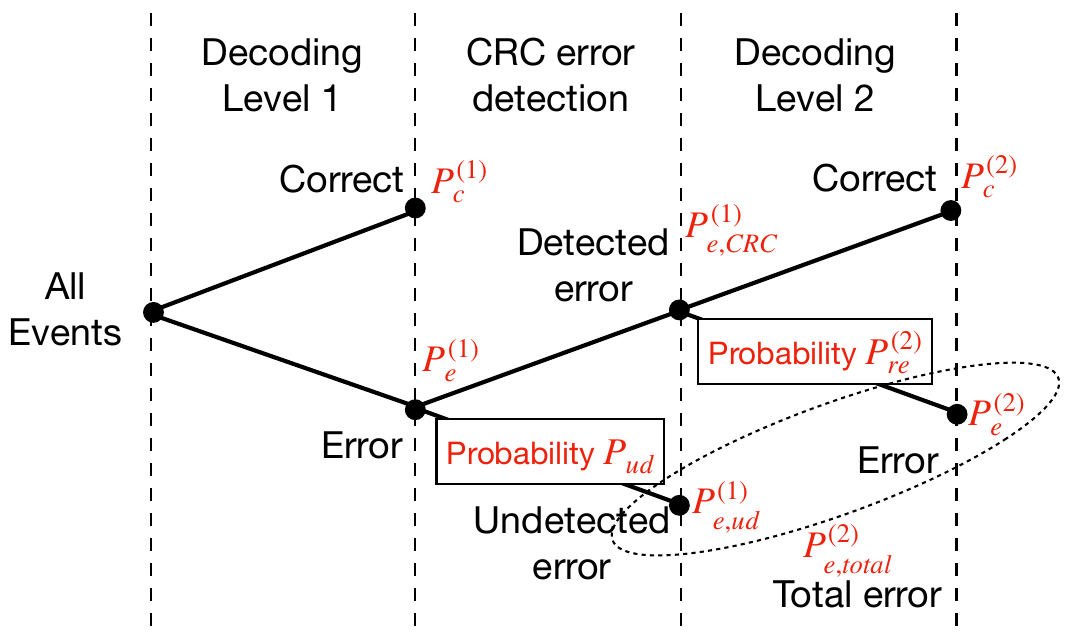}
    \caption{Events and probabilities for retry decoding with 2 levels.}
    \label{fig_error_event_retry_ana}

\end{figure}

{{}
Embedding longer CRC gives better error detection capability, while larger SNR penalty is suffered. This section considers this trade-off to balance error detection capability and SNR penalty. A semi-analytic CRC length optimization is derived to maximize the SNR gain, involving the SNR penalty, at a given target error rate. This optimization is valid for both single user transmission and CF relaying. 
}

Table~\ref{table_p_ud} showed that a longer CRC has the better error detection capability while a larger SNR penalty \eqref{equ_SNR_penalty} is suffered. The CRC length optimization balances this trade-off to maximize the gain for a target WER/EER. The optimization consists of two steps: first, estimate the WER/EER after $k$-level retry decoding, denoted as $P_{e, total}^{(k)}$, as a function of CRC length $l$; then obtain the gain from the WER/EER curve using the estimated $P_{e, total}^{(k)}$ with the SNR penalty included. The optimal $l$ gives the lowest SNR for a given target WER/EER. The $P_{e, total}^{(k)}$ is first derived as a function of the probability of undetected error $P_{ud}$ defined in \eqref{equ_def_p_ud}: 
\begin{align} \label{equ_P_e_total_func}
    f: P_{ud} \rightarrow P_{e, total}^{(k)}.
\end{align}
Using \eqref{equ_p_ud_2}-\eqref{equ_p_ud_esti_l}, define a mapping between CRC length $l$ and $P_{ud}$, $g: l \rightarrow P_{ud}$. Then $P_{e, total}^{(k)}$ given in \eqref{equ_P_e_total_func} is further written as:
\begin{align}
    f \circ g: l \rightarrow P_{e, total}^{(k)}.
\end{align}
Notations for deriving $f: P_{ud} \rightarrow P_{e, total}^{(k)}$ are defined as follows. Let $P_{e}^{(k)}$ be the word error probability at the $k$-th decoding level. After decoding, a CRC check splits these error events into detected and undetected errors with probability $P_{e, CRC}^{(k)}$ and $P_{e, ud}^{(k)}$, respectively, where events with probability $P_{e, ud}^{(k)}$ are not retried in future decoding. The total word error probability $P_{e, total}^{(k)}$ for $k>1$ is:
\begin{align} \label{equ_P_e_total_k}
    P_{e, total}^{(k)}= \sum_{i= 1}^{k- 1} P_{e, ud}^{(i)}+ P_{e}^{(k)},
\end{align} 
due to the mutually exclusivity of events.
For $k= 1$, $P_{e, total}^{(1)}= P_{e}^{(1)}$ is the word error probability of the one-shot decoder. For $i= 2, 3, \cdots, k$, another new term $P_{re}^{(i)}$ is defined to indicate the word error probability at the $i$-th decoding level given a detected error from level $i- 1$. 
Fig.~\ref{fig_error_event_retry_ana} illustrates an example of the structure of events and corresponding probabilities for a 2-level decoding.

\begin{myProposition} \label{prop_p_e_re}
    \rm  For decoding with $k> 1$ levels, given $P_{e}^{(1)}$ and $P_{re}^{(i)}$, for $i= 2, 3, \cdots, k$, of the CRC-embedded lattice code, the function $f: P_{ud} \rightarrow P_{e, total}^{(k)}$ is expressed as:
    \begin{align} \label{equ_P_e_re_estim_k}
        P_{e, total}^{(k)} & = \sum_{i= 1}^{k-1} \left( P_{e}^{(i)} P_{ud} \right)+ P_{re}^{(k)} P_{e}^{(k- 1)} (1- P_{ud}),
    \end{align}
    where $P_{e}^{(i)}$ for $i> 1$ is obtained recursively as:
    \begin{align} \label{equ_p_e_rec}
        P_{e}^{(i)}= P_{re}^{(i)} P_{e}^{(i-1)} (1- P_{ud}).
    \end{align}
\end{myProposition}

\begin{proof}
    By the structure shown in Fig.~\ref{fig_error_event_retry_ana}, $P_{e, CRC}^{(i)}$ can be estimated on average as $P_{e, CRC}^{(i)}= P_{e}^{(i)} (1- P_{ud})$, by which $P_{e}^{(i)} (i> 1)$ can be recursively obtained as: 
    \begin{align} \label{equ_p_e_rec_proof}
        P_e^{(i)} = P_{re}^{(i)} P_{e, CRC}^{(i)}=  P_{re}^{(i)} P_{e}^{(i-1)} (1- P_{ud}).
    \end{align}
    Similarly, we have $P_{e, ud}^{(i)}= P_{e}^{(i)} P_{ud}$. Substituting $P_{e, ud}^{(i)}$ and $P_e^{(i)}$ into \eqref{equ_P_e_total_k}, it is obtained that:
    \begin{align} \label{equ_p_e_total_fun}
        P_{e, total}^{(k)} = \sum_{i= 1}^{k-1} \left( P_{e}^{(i)} P_{ud} \right)+ P_{re}^{(k)} P_{e}^{(k- 1)} (1- P_{ud}).
    \end{align}
    With given $P_{e}^{(1)}$ and $P_{re}^{(i)}$ for $i= 2, 3, \cdots, k$, we can see \eqref{equ_p_e_total_fun} as a function of $P_{ud}$, that is, $f: P_{ud} \rightarrow P_{e, total}^{(k)}$.
\end{proof}

{{}
The CRC length optimization is a semi-analytical method because estimating using $f$ requires numerically evaluated $P_{e}^{(1)}$ and $P_{re}^{(i)}$, $i= 2, 3, \cdots, k$. Fortunately, the values of $P_{e}^{(1)}$ and $P_{re}^{(i)}$, for $i= 2, 3, \cdots, k$ depend on the value of decoding coefficients but are independent of the embedded CRC code, thus numerical evaluation over different CRC length $l$ is not needed. To justify this, recall that the lattice decoder uses the base lattice but not the CRC-embedded lattice. First, since $P_{e}^{(1)}$ is the word error probability of the conventional one-shot decoding so that it is independent of the embedded CRC code. The lowest $P_{e}^{(1)}$ is achieved by selecting the optimal decoding coefficient, such as $\alpha_{MMSE}$ for single user transmission. Second, $P_{re}^{(i)}$ is the transition probability between \emph{a detected error} with probability $P_{e, CRC}^{(i- 1)}$ and \emph{a decoding error of next decoding attempt} with probability $P_{e}^{(i)}$, see Fig.~\ref{fig_error_event_retry_ana} for $i= 2$. It can be seen that the measure of $P_{re}^{(i)}$ starts from an error that is already detected; passes through the retry decoding using the base lattice decoder; then ends at an error event before error detection which includes both detected and undetected errors at the $i$-th error detection. Since no knowledge of the CRC code is involved during this transition, the probability $P_{re}^{(i)}$ is independent of the embedded CRC codes, as well as the CRC length $l$. Similar to $P_{e}^{(1)}$, the value of $P_{re}^{(i)}$ also depends on the choice of retry decoding coefficient, see Fig.~\ref{fig_E8_alphasub_condi} as an example for single user transmission. 

Numerically, $P_{re}^{(i)}$ can be evaluated using
\begin{align}  \label{equ_P_re_k_CRC}
    P_{re}^{(i)}= P_e^{(i)}/ P_{e, CRC}^{(i- 1)}
\end{align}
by embedding arbitrary CRC code for error detection, or equivalently using
\begin{align} \label{equ_P_re_k_genie}
    P_{re}^{(i)}= P_e^{(i)}/ P_{e}^{(i- 1)}
\end{align}
by assuming genie-aided error detection with $P_{ud}= 0$. 

By letting $P_{ud}= 0$, the estimate of $P_{e, total}^{(k)}= P_{re}^{(k)} P_{e}^{(k- 1)}$ indicating the word error probability using genie-aided error detection, which gives a lower bound of WER/EER for $k$-level retry decoding with no SNR penalty included.
}

{{}
\section{Implementation of CRC-embedded lattice codes} \label{sec_implement}
This section considers implementation of the CRC-embedded lattice codes for both single user transmission and CF relaying scenarios with optimized CRC length. First, trade-offs related to implementation are discussed. Then numerical results are provided using $E_8$ and $BW_{16}$ lattice codes for single user transmission through the AWGN channel; and $N= 128, 256$ polar lattice codes for a 2-user CF relay through Rayleigh fading channel, respectively. Since channel values are random for CF relaying, scheme 1 in Section~\ref{sec_mac_CFrelay_dec_scheme} is applied as the retry decoding scheme.

\subsection{Trade-offs on implementation} \label{sec_implement_tradeoff}
In the previous section, we addressed the trade-off between CRC error detection capability and SNR penalty, and a CRC length optimization was given. Here we introduce some other trade-offs related to implementation of the proposed CRC-embedded lattices/lattice codes. 

The first trade-off is between lattice dimension and SNR penalty for fixed CRC length $l$ and code rate $R$ of the base lattice code $\mathcal{C}= \Lambda_c/\Lambda_s$, rather than the CRC-embedded lattice code. The SNR penalty in \eqref{equ_SNR_penalty} can be written as $SNR_p= 10 \log_{10} (1+ l/(NR- l))$. For lattice codes with different dimension $N$ but the same rate $R$, the SNR penalty decreases when $N$ increases. Also, by definition, the increase of SNR penalty is more significant for low dimensional lattice codes when adding 1 CRC bit. This implies that for low dimensional lattice codes, short CRC length is preferred to achieve small SNR penalty by sacrificing error detection capability. However, when $N$ increases, longer CRC can be applied to have better error detection capability without increasing too much SNR penalty.

The second trade-off is between code rate $R$ of the base lattice code $\mathcal{C}= \Lambda_c/\Lambda_s$ and the improvement of retry decoding for fixed dimension $N$ and CRC length $l$. When code rate $R$ increases, the fraction of the number of CRC parity bits over total number of bits is smaller so that SNR penalty is reduced. This implies that for an $N$-dimensional lattice code, one possible direction to increase the improvement of retry decoding is to expand the shaping region to increase the code rate.

Another notable trade-off is between lattice dimension and improvement of retry decoding. From a theoretical point of view, as $N \rightarrow \infty$, there exists a lattice code achieves the Gaussian channel capacity by using $\alpha_{MMSE}$ \cite{erez2004achieving} in single user transmission and a lattice code achieves Poltyrev's bound by using coefficient set $\{\mathbf{a}, \alpha\}$, which maximizes the computation rate \cite{nazer2011compute}, in CF relaying. This implies the improvement of retry decoding decreases as $N \rightarrow \infty$. For the single user case, as described in Section~\ref{sec_su}, the only noise component is the additive white Gaussian noise. As dimension $N$ increases, it is known that the probability density of Gaussian noise is concentrated in a thin annulus with radius $\sqrt{N \sigma^2}$ \cite{blum2020foundations}, i.e. exist $\epsilon > 0$ such that
\begin{align}
    \Pr\left(\sqrt{N \sigma^2}- \epsilon< \|\mathbf{z}\|< \sqrt{N \sigma^2}+ \epsilon \right) \rightarrow 1.
\end{align}
In single user transmission, by letting the Voronoi region $\mathcal{V}$ cover a sphere with radius $\sqrt{N \sigma^2}/ \alpha_{MMSE}$, the improvement of retry decoding decreases as $N$ increases. 

For CF relaying described in Section~\ref{sec_mac_CFrelay}, retry decoding changes the equivalent noise \eqref{equ_CF_message_equi_noise} in the way of: both the length and direction of integer approximation error term $\sum_{i= 1}^L (\alpha h_{i}- a_{i}) \mathbf{x}_i$, and the length of the Gaussian noise term $\alpha \mathbf{z}$ with direction unchanged. The change of direction of integer approximation error may lead to a vector cancellation against Gaussian noise to correct an error in actual decoding. 

Overall, the SNR penalty can be reduced by increasing lattice dimension $N$ or code rate $R$ of the base lattice code $\mathcal{C}= \Lambda_c/\Lambda_s$. However, even though lower SNR penalty can be achieved, increasing $N$ also reduces the improvement of retry decoding. It becomes important to select an appropriate dimension for implementing the CRC-embedded lattice codes.
Numerically, for single user transmission, improvement of retry decoding using the CRC-embedded lattices can be observed for low dimensional lattice codes and becomes marginal for medium dimensional lattice codes. For example, as shown in \cite[Fig.~7]{xue2022lower}, even though an SNR penalty is as small as 0.078dB by a $N= 128$ polar lattice code with CRC-4 embedded, the gain is still marginal using three $\alpha$ candidates for retry decoding. For CF relaying, due to the existence of integer approximation error term, the improvement of retry decoding can still be observed for medium dimensional lattice codes. To implement the CRC-embedded lattice codes, low dimension lattice codes, such as $E_8$ and $BW_{16}$ lattice codes, are considered in single user transmission; and medium dimensional lattice codes, such $N= 128, 256$ polar lattice codes, are considered in CF relaying. Since the SNR penalty is significant for low dimensional lattice codes, high code rates are also considered in single user case, which can be seen as high-order modulations in communications.
}

\subsection{Numerical results for single user transmission} \label{sec_implement_SU}
Next we give an implementation of retry decoding and CRC-embedded lattice codes for single user transmission. The decoding procedure follows Section~\ref{sec_SU_dec_scheme} using a list $\mathcal{A}_1, \mathcal{A}_2, \cdots, \mathcal{A}_k$. Here, CRC codes are applied for error detection instead of the genie used in the analysis, thus the SNR penalty is non-negligible. After each decoding attempt, the uncoded message $\hat{\mathbf{b}}$ is recovered by the indexing function in \eqref{equ_latticecode_index} from which the LSB vector is extracted for CRC check.

To estimate $P_{e, total}^{(k)}$, the value of $P_{e}^{(1)}$ is evaluated by one-shot decoding using $\alpha_{MMSE}$. By the recursive structure of retry decoding, $P_{e}^{(i)}$ can be calculated using $P(\alpha | e_{i- 1})$ defined in \eqref{equ_Prob_previous_failed} as:
\begin{align}
    P_e^{(i)} = P_{e, CRC}^{(i- 1)}- \sum_{j= 1}^{2^{i- 1}} P(\alpha_{i, j} | e_{i- 1}) P_{e, CRC}^{(i- 1)}.
\end{align}
By \eqref{equ_P_re_k_CRC}, $P_{re}^{(i)}$ can then be obtained without requiring extra numerical evaluation as:
\begin{align}
    P_{re}^{(i)}= 1- \sum_{j= 1}^{2^{i- 1}} P(\alpha_{i, j} | e_{i- 1}).
\end{align}

The numerical results are given using $E_8$ and $BW_{16}$ lattice codes with hypercube shaping. Since $E_8$ and $BW_{16}$ lattice codes have known kissing number, the CRC polynomials are selected to minimize $P_{ud}$ estimated using \eqref{equ_p_ud_esti_kiss}. With the SNR penalty in \eqref{equ_SNR_penalty} included, Fig.~\ref{fig_E8_BW16_Petotal_esti_simu} verifies the accuracy of the estimate using Proposition~\ref{prop_p_e_re} for single user transmission. For target WER$=10^{-5}$ and 2-level decoding, the best achievable gain and the optimized CRC length $l$ for $E_8$ and $BW_{16}$ lattice codes are shown in Table~\ref{table_E8_rec_CRC} and Table~\ref{table_BW16_rec_CRC}, where $R$ in tables are the code rates of based lattice code $\mathcal{C}= \Lambda_c/ \Lambda_s$ before CRC embedding. 
{{} Regarding to the first and second trade-offs discussed in Section~\ref{sec_implement_tradeoff}, the optimized CRC lengths are short, since low dimensional lattice codes are applied; while, high rate code are considered to reduce the SNR penalty to achieve larger gain.}
When $R= 2,3,4$ for $E_8$ lattice codes and $R= 2.25$ for $BW_{16}$ lattice codes, the expected gain is less than 0 for which embedding parity bits and retry decoding are not recommended. 

\begin{figure}[t]
    \centering
    \includegraphics[width=0.9\linewidth]{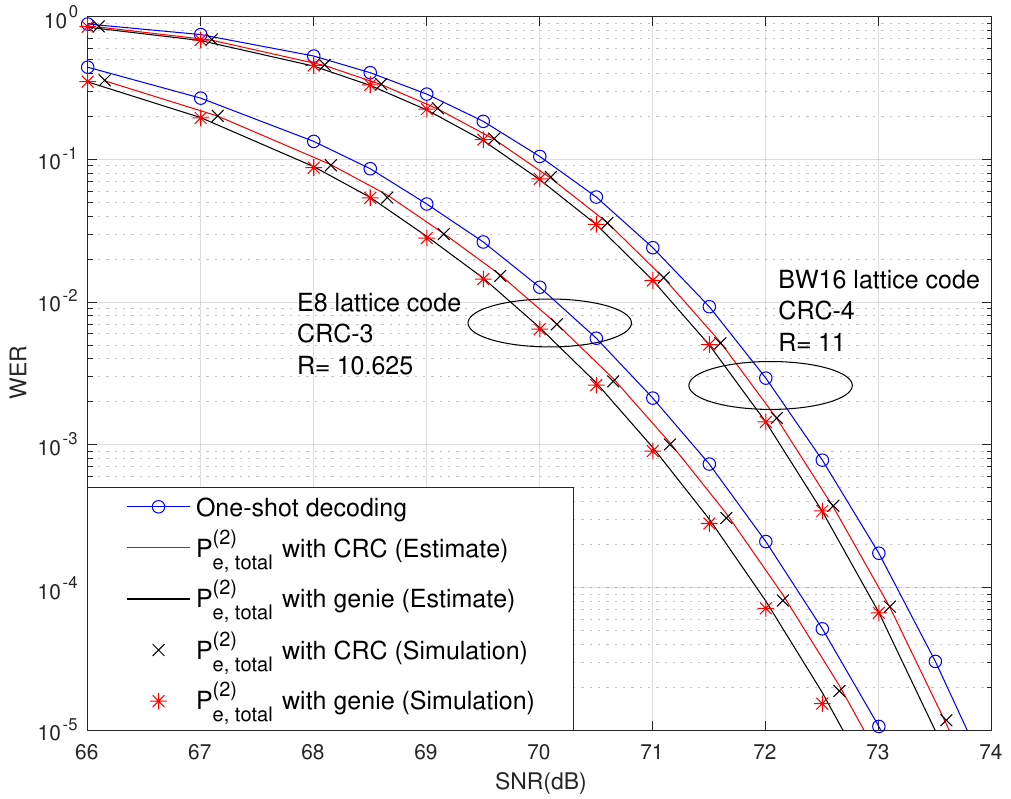}
    \caption{Estimated and actual $P_{e, total}^{(2)}$ for single user transmission using $E_8$ and $BW_{16}$ lattice codes. The CRC polynomials are $x^3+ x+ 1$ and $x^4+ x^3+ 1$, respectively. The total decoding level is 2 using $\mathcal{A}_1= \{\alpha_{MMSE}\}$, $\mathcal{A}_2= \{\alpha_{2, 1}, \alpha_{2, 2}\}$.}
    \label{fig_E8_BW16_Petotal_esti_simu}
\end{figure}

\begin{table}[t]
    \centering
    \caption{The expected gain and optimized CRC length $l$ for $E_8$ lattice code with hypercube shaping and 2 decoding levels.}
    \begin{tabular}{|c|c|c|c|}
        \hline
        $R$ & Gain (dB) & Optimized $l$ & CRC polynomial   \\ \hline
        2, 3, 4 & $< 0$  & -                  & -                                  \\ \hline
        5       & 0.0060 & 1                  & SPC code                           \\ \hline
        6       & 0.0352 & \multirow{2}{*}{2} & \multirow{2}{*}{$x^2+ 1$}          \\ \cline{1-2}
        7       & 0.0621 &                    &                                    \\ \hline
        8       & 0.0845 & \multirow{4}{*}{3} & \multirow{4}{*}{$x^3+ x+ 1$}       \\ \cline{1-2}
        9       & 0.1082 &                    &                                    \\ \cline{1-2}
        10      & 0.1270 &                    &                                    \\ \cline{1-2}
        11      & 0.1424 &                    &                                    \\ \hline
        -       & 0.3270 & \multicolumn{2}{|c|}{Upper bound on the gain.} \\ \hline
    \end{tabular}
    \label{table_E8_rec_CRC}
\end{table}

\begin{table}[t]
    \centering
    \caption{The expected gain and optimized CRC length for $BW_{16}$ lattice code with hypercube shaping and 2 decoding levels.}
    \begin{tabular}{|c|c|c|c|}
        \hline
        $R$ & Gain (dB) & Optimized $l$ & CRC polynomial  \\ \hline
        2.25  & $< 0$  & -                  & -                                   \\ \hline
        3.25  & 0.0197 & 1                  & SPC code                            \\ \hline
        4.25  & 0.0484 & 2                  & $x^2+ 1$                            \\ \hline
        5.25  & 0.0741 & \multirow{5}{*}{3} & \multirow{5}{*}{$x^3+ x^2+ 1$}      \\ \cline{1-2}
        6.25  & 0.0997 &                    &                                     \\ \cline{1-2}
        7.25  & 0.1182 &                    &                                     \\ \cline{1-2}
        8.25  & 0.1322 &                    &                                     \\ \cline{1-2}
        9.25  & 0.1431 &                    &                                     \\ \hline
        10.25 & 0.1528 & \multirow{2}{*}{4} &  \multirow{2}{*}{$x^4+ x^3+ 1$}     \\ \cline{1-2}
        11.25 & 0.1624 &                    &                                     \\ \hline
        -     & 0.2880 & \multicolumn{2}{|c|}{Upper bound on the gain.} \\ \hline
    \end{tabular}
    
    \label{table_BW16_rec_CRC}
\end{table}

\subsection{Numerical results for CF relaying}  \label{sec_implement_CF}
{{} For implementation in CF relaying using ICF, the Rayleigh fading channel is assumed. The decoding procedure follows scheme 1 in Section~\ref{sec_mac_CFrelay_dec_scheme} but a CRC code is applied for error detection instead of a genie. Let $\mathcal{C}= \Lambda_c/ \Lambda_s$ be the base lattice code before CRC embedding and $\mathbf{x}_1, \mathbf{x}_2 \in \mathcal{C}$. The received message is $\mathbf{y}= h_1 \mathbf{x}_1+ h_2 \mathbf{x}_2+ \mathbf{z}$, where the total channel gain is normalized to $\|\mathbf{h}\|= \sqrt{h_1^2+ h_2^2}= 1$, to keep a constant received SNR. Using ICF, the estimate of the linear combination is $\hat{\mathbf{x}} \in \Lambda_c$ (not necessarily in $\mathcal{C}$). The corresponding uncoded message is recovered by $\hat{\mathbf{b}}= \mathbf{G}_c^{-1} \hat{\mathbf{x}}$, from which the LSB vector is extracted for CRC check. To compute $P_{e, total}^{(k)}$ in \eqref{equ_P_e_re_estim_k} for CRC length optimization, $P_{ud}$ is estimated by \eqref{equ_p_ud_esti_l} as $P_{ud} \approx 2^{-l}$, and both $P_{e}^{(1)}$ and $P_{re}^{(i)}$, for $i= 2, 3, \cdots, k$, are numerically evaluated using Monte-Carlo method. The equation error rate (EER) of the linear combinations is measured at the relay, when $\hat{\mathbf{x}} \neq \sum_{i= 1}^L a_{i} \mathbf{x}_i$.

}

Numerical evaluation of retry decoding is given for a 2-user CF relay using ICF. Construction D polar lattice codes with $N= 128, 256$ and hypercube shaping are used for channel coding. The parameters for lattice design follow \cite{ludwiniananda2021design}. {{} The standard successive cancellation (SC) decoder is applied for the component polar codes. It is noticed that polar codes can use CRC codes for error detection in successive cancellation list (SCL) decoding to improve error performance compared with SC decoding \cite{tal2015list}. The CRC for SCL decoding performs error detection at the component binary codeword level; while the proposed CRC-embedded lattice code performs error detection at the lattice codeword level. Even though SCL achieves better error performance than SC on decoding polar codes, to avoid ambiguity between two types of CRC codes, the standard SC decoder is considered.} 

Fig.~\ref{fig_N128_256_opt_CRC4_esti_simu} verifies the accuracy of the estimate in Proposition~\ref{prop_p_e_re} for CF relaying with CRC-4 and genie-aided error detection, where two decoding attempts are assumed. For the genie-aided case, the users' messages $\mathbf{x}_i$ are available at relay and the error detection is performed by comparing $\hat{\mathbf{y}}$ with the true linear combination $\sum_{i= 1}^2 a_{i} \mathbf{x}_i$. At target EER $=10^{-5}$, gains of approximately $1.51$ dB and $1.18$ dB are observed with genie-aided error detection, indicating the upper bound on gain for retry decoding under the configuration above. Using CRC-embedded lattice codes, the expected gain for $N= 128, 256$ polar lattice codes and retry decoding with two attempts are shown in Fig.~\ref{fig_N128_256_CRClength_gain} for CRC length $l= 1, 2, \cdots, 16$. The maximum gains of approximately $1.31$ dB and $ 1.08$ dB are achieved when the CRC length are in the range of 8-9 and 9-11 for $N= 128$ and $256$, respectively. For longer CRC lengths, the gain decreases because the increasing SNR penalty overcomes the error performance improvement. 

\begin{figure}[t]
    \centering
    \includegraphics[width=0.9\linewidth]{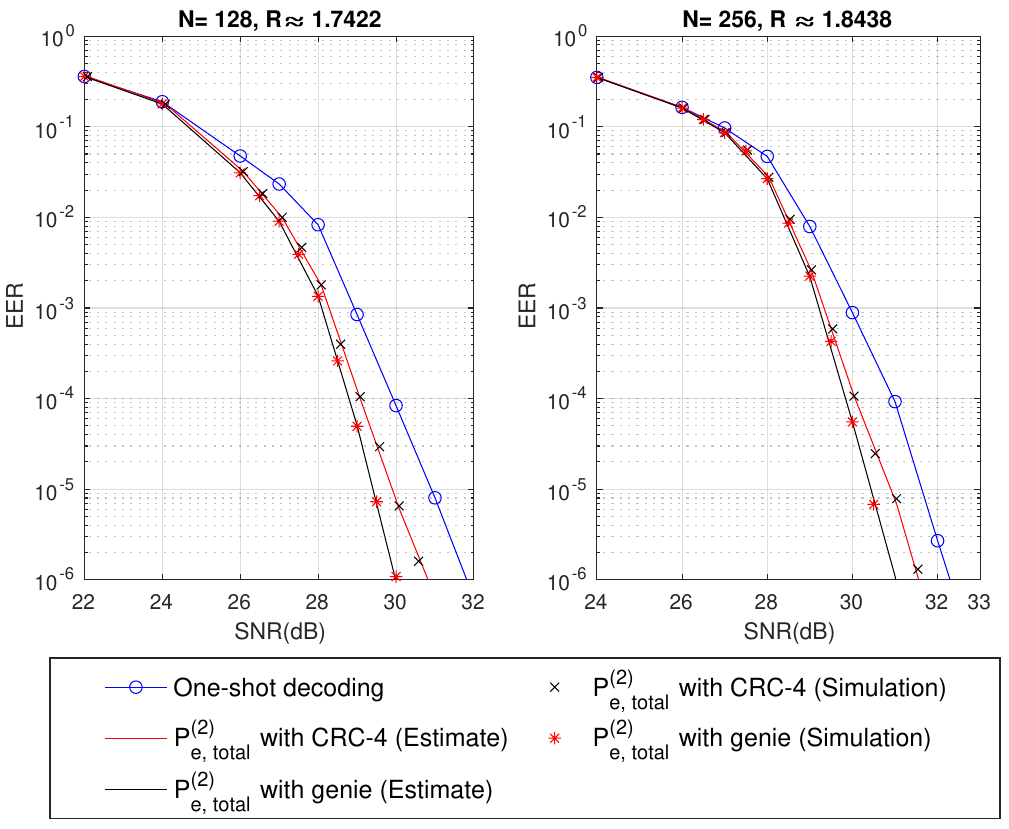}
    \caption{Estimate and simulation results of $P_{e, total}^{(2)}$ for 2-user CF relay using ICF with CRC-4 and genie-aided error detection. $N= 128, 256$ polar code lattice and hypercube shaping is used.}
    \label{fig_N128_256_opt_CRC4_esti_simu}
\end{figure}

\begin{figure}[t]
    \centering
    \includegraphics[width=0.9\linewidth]{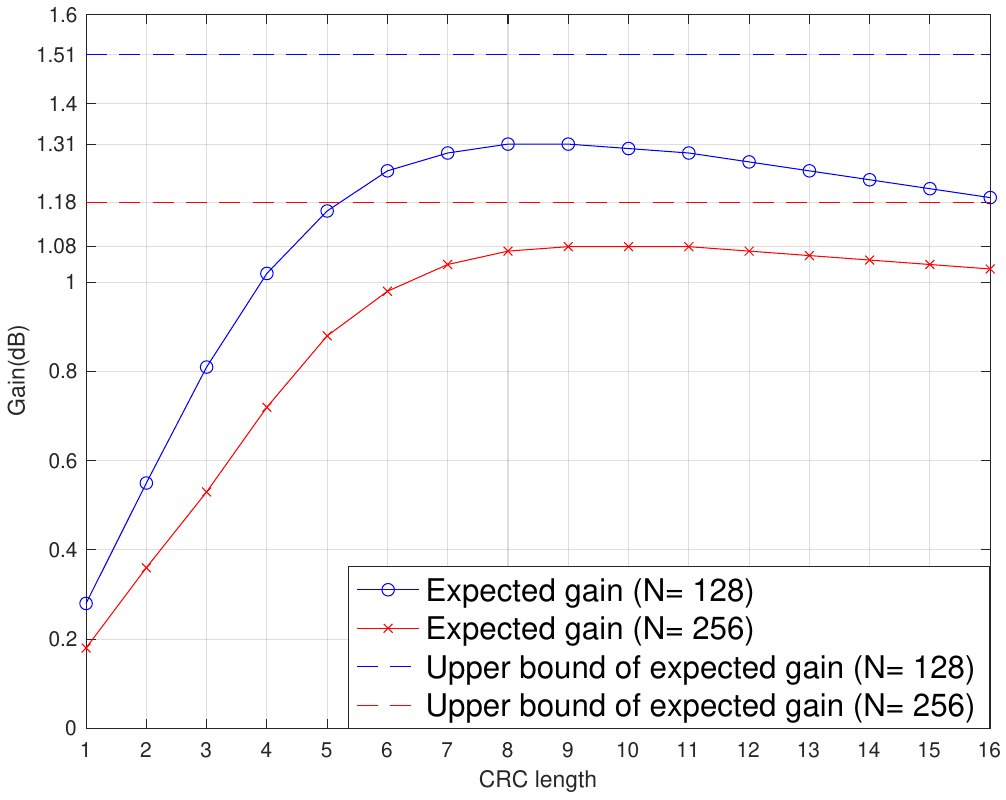}
    \caption{Expected gain for 2-user CF relay using $N= 128, 256$ polar code lattice with CRC length from 1 to 16. For $N= 128$, the maximum gain of $1.31$dB is achieved with CRC length being $8,9$; for $N= 256$, the maximum gain of $1.08$dB is achieved with CRC length being $9, 10, 11$.}
    \label{fig_N128_256_CRClength_gain}
\end{figure}

\section{Conclusions} \label{sec_conclusions}
This paper considered finite dimensional lattice-based communications for point-to-point single user transmission and multiple access relaying using compute-forward. It is shown that the proposed CRC-embedded lattice codes with retry decoding outperforms the conventional one-shot decoding with a lower decoding error rate. For CF relaying, this scheme is applicable if the relay is power unconstrained and the shaping lattice is appropriately designed, by which an error in linear combinations can be detected without requiring individual users' messages. For the aspect of practical lattice code design, a semi-analytical method to optimize the CRC length was introduced. The implementation examples illustrated that the proposed scheme can improve the error performance for both single user transmission and CF relaying using ICF, where the gains were more significant for CF relaying than the single user case. For a 2-user relay, gains up to $1.31$ dB and $1.08$ dB at EER$\approx 10^{-5}$ are observed by only adding one more decoding attempt for $N= 128, 256$ Construction D polar lattice codes, respectively, with the optimized CRC length. 

{{} The proposed retry decoding repeats lattice decoding for base lattice, rather than the CRC-embedded lattice, $k$ times, where $k$ is the number of decoding attempts and can be seen as a constant compared with the complexity of base lattice decoder. The overall complexity of decoding maintains the same level as the one-shot decoding which is dominated by the base lattice decoder. For small WER, such as $P_e< 10^{-5}$, most of codewords are successfully recovered in the first decoding attempt, resulting in the average number of decoding attempts being a value close to 1.}
In addition, error detection only requires $\bmod\ 2$ operations and a CRC check, thus lower time latency can be expected by applying CRC-embedded lattice codes and retry decoding, compared with requesting a re-transmission. Also, for CF relaying, if the error can be detected from linear combinations, the relay can stop forwarding erroneously decoded messages into the network. 

{{} Some potential applications and extensions of this work are considered. For single user transmission, we considered low dimensional lattice codes with high code rates in Section~\ref{sec_implement_SU}. Since lattice codes can be seen as a coded modulation scheme with shaping gain achieved, it can be considered as a competing scheme with QAM to reduce the constellation power. High rate lattice codes correspond to high-order modulations, such as 1024-QAM or 4096-QAM which are considered in fiber optic systems \cite{chen2015real}, the IEEE 802.3ax (Wi-Fi 6) \cite{WiFi6} and IEEE 802.3be (Wi-Fi 7) standards \cite{WiFi7}. While even higher order modulation is considered for future standards \cite{reshef2022future} suggesting the need for high rate codes.}
For CF relaying, first is to study potential applications of CRC-embedded lattice codes in practical scenarios, such as mobile communications and sensor networks. Considering the links connecting relay nodes and destination as wired back-haul links which have no transmit power constraint, the ICF scheme can be applied to achieve lower forwarding latency than conventional decode-forward. Meanwhile, CRC-embedded lattice codes provide physical layer error detection to reduce error rates and prevent forwarding erroneously decoded messages into networks. 
Another direction is to study the improvement of network throughput and error performance at the destination in a CF network, where retry decoding is performed at relays. In this case, the CF integer coefficients selection at the various relays should be considered in order to have a full rank coefficient matrix at the destination. {{}Additionally, MIMO integer-forcing (IF) receiver \cite{zhan2014integer} also applies lattice codes for PLNC which is extended from CF relaying. Unlike a standalone relay in CF network, the IF receiver has all linear combinations and can solve them locally, therefore it is possible to omit the $\bmod \Lambda_s$ operation during decoding. An application of the proposed CRC-embedded lattice codes with retry decoding to IF receiver can also be considered to reduce the error rate, as another extension of this work. }

\footnotesize
	\bibliographystyle{ieeetr}
	\bibliography{Reference}

\newpage

\vfill

\end{document}